%% file: main.tex
\documentclass{article}

\usepackage[final,nonatbib]{neurips_2022}
\usepackage[numbers]{natbib}

\usepackage[utf8]{inputenc} 
\usepackage[T1]{fontenc}    
\usepackage{hyperref}       
\usepackage{url}            
\usepackage{booktabs}       
\usepackage{amsfonts}       
\usepackage{nicefrac}       
\usepackage{microtype}      
\usepackage{xcolor}         
\usepackage{amsmath}
\usepackage{mathtools}
\usepackage{amsthm}
\usepackage{amsmath}
\usepackage{amssymb}
\usepackage{bbm}
\usepackage{tikz} 
\usepackage{wrapfig}
\usepackage{floatrow}
\usepackage{enumitem}
\usepackage{caption}
\usepackage{subcaption}
\usepackage{booktabs} 
\usepackage[ruled]{algorithm2e} 

\usepackage[suppress]{color-edits}
\addauthor[Hoda]{hh}{cyan}
\addauthor{kh}{red}
\addauthor{vc}{blue}
\addauthor{jk}{orange}
\addauthor{sw}{blue}

\title{Bayesian Persuasion for Algorithmic Recourse}

\author{Keegan Harris\\
Carnegie Mellon University\\
\texttt{keeganh@cmu.edu}\\
\And
Valerie Chen\\
Carnegie Mellon University\\
\texttt{valeriechen@cmu.edu}\\
\And
Joon Sik Kim\\
Carnegie Mellon University\\
\texttt{joonkim@cmu.edu}\\
\And
Ameet Talwalkar\\
Carnegie Mellon University\\
\texttt{talwalkar@cmu.edu}\\
\And
Hoda Heidari\\
Carnegie Mellon University\\
\texttt{hheidari@cmu.edu}\\
\And
Zhiwei Steven Wu\\
Carnegie Mellon University\\
\texttt{zstevenwu@cmu.edu}
}

\begin{document}
\newtheorem{theorem}{Theorem}[section]
\newtheorem{lemma}[theorem]{Lemma}
\newtheorem{fact}[theorem]{Fact}
\newtheorem{claim}[theorem]{Claim}
\newtheorem{remark}[theorem]{Remark}
\newtheorem{corollary}[theorem]{Corollary}
\newtheorem{prop}[theorem]{Proposition}
\newtheorem{assumption}[theorem]{Assumption}
\newtheorem{example}[theorem]{Example}
\newtheorem{definition}[theorem]{Definition}

\newcommand{\vtheta}{\boldsymbol{\theta}}
\newcommand{\vTheta}{\boldsymbol{\Theta}}
\newcommand{\vgamma}{\boldsymbol{\gamma}}
\newcommand{\vsigma}{\boldsymbol{\sigma}}
\newcommand{\vx}{\mathbf{x}}
\newcommand{\va}{\mathbf{a}}
\newcommand{\vs}{\mathbf{s}}
\newcommand{\xhdr}[1]{\textbf{#1.}}

\maketitle

\begin{abstract}
When subjected to automated decision-making, decision subjects may strategically modify their observable features in ways they believe will maximize their chances of receiving a favorable decision. In many practical situations, the underlying assessment rule is deliberately kept secret to avoid gaming and maintain competitive advantage. The resulting opacity forces the decision subjects to rely on \emph{incomplete information} when making strategic feature modifications. We capture such settings as a game of \emph{Bayesian persuasion}, in which the decision maker offers a form of recourse to the decision subject by providing them with an action recommendation (or \emph{signal}) to incentivize them to modify their features in desirable ways. We show that when using persuasion, the decision maker and decision subject are \emph{never worse off} in expectation, while the decision maker can be \emph{significantly better off.} While the decision maker's problem of finding the optimal Bayesian incentive-compatible (BIC) \emph{signaling policy} takes the form of optimization over infinitely-many variables, we show that this optimization can be cast as a linear program over finitely-many regions of the space of possible assessment rules. While this reformulation simplifies the problem dramatically,
solving the linear program requires reasoning about exponentially-many variables, even in relatively simple cases. Motivated by this observation, we provide a polynomial-time approximation scheme that recovers a near-optimal signaling policy. Finally, our numerical simulations on semi-synthetic data 
empirically demonstrate the benefits of using persuasion in the algorithmic recourse setting.
\end{abstract}

\input{intro}
\input{related}
\input{prelims}
\input{model-short}
\input{experiments}
\input{conclusion}
\input{acks}

\newpage 
\bibliographystyle{abbrvnat}
\bibliography{refs}
\newpage
\input{checklist}

\appendix
\input{social}
\input{exp-compute}
\input{example_appendix}
\input{computational_barriers}
\input{approx_appendix_full}
\input{exp_appendix}

\end{document}

%% file: intro.tex

\section{Introduction}
High-stakes decision-making systems increasingly utilize data-driven algorithms to assess individuals in such domains as education~\citep{kuvcak2018machine}, employment~\citep{chalfin2016productivity,raghavan2020mitigating}, and lending~\citep{jagtiani2019roles}. Individuals subjected to these assessments (henceforth, decision subjects) may strategically modify their observable features in ways they believe maximize their chances of receiving favorable decisions \cite{homonoff2021does, citron2014scored}. The decision subject often has a set of actions/interventions available to them. Each of these actions leads to some measurable effect on their observable features, and subsequently, the decision they receive. From the decision maker's perspective, some of these actions may be more desirable than others. 
Consider credit scoring as an example.\footnote{Other examples of strategic settings which arise as a result of decision-making include college admissions, in which a college/university decides whether or not to admit a prospective student, hiring, in which a company decides whether or not to hire a job applicant, and lending, in which a banking institution decides to accept or reject someone applying for a loan. Oftentimes, the decision maker is aided by automated decision-making tools in these situations (e.g., \cite{kuvcak2018machine, sanchez2020does, jagtiani2019roles}).\looseness-1}  Credit scores predict how likely an individual applicant is to pay back a loan on time. Financial institutions regularly utilize credit scores to decide whether to offer applicants their financial products and determine the terms and conditions of their offers (e.g., by setting the interest rate or credit limit). Applicants regularly attempt to improve their scores given their (partial) knowledge of credit scoring instruments. For instance, a business applying for a loan may improve its score by paying off existing debt or cleverly manipulating its financial records to appear more profitable. While both of these interventions may improve credit score, the former is more desirable than the latter from the perspective of the financial institution offering the loan. The question we are interested in answering in this work is: \emph{how can the decision maker incentivize decision subjects to take such beneficial actions while discouraging manipulations?}\looseness-1

The strategic interactions between decision-making algorithms and decision subjects has motivated a growing literature known as \emph{strategic learning} (see e.g., \cite{hardt2016strategic, dongetal, shavit2020causal, kleinberg2020classifiers, harris2021stateful}).  While much of the prior work in strategic learning operates under the assumption of full transparency (i.e., the assessment rule is public knowledge), we consider settings where the full disclosure of the assessment rule is not viable. In many real-world situations, revealing the exact logic of the decision rule is either infeasible or irresponsible. For instance, credit scoring formulae are closely guarded trade secrets, in part to prevent the risk of default rates surging if applicants learn how to manipulate them. Moreover, the underlying decision rule is often \emph{fixed} ahead of time due to institutional structuring. In our credit scoring example, one department of the bank may be in charge of determining the threshold on the credit assessment, while another department may be in charge of offering recourse.\footnote{Similar logic applies to other examples of strategic settings including: college admissions, in which someone associated with the university may have the ability to offer advice to applicants, but does not have the ability to unilaterally change the underlying assessment rule, or hiring, where a recruiter for a company may have knowledge of the factors the company uses to make hiring decisions, but may not be able to change this criteria or reveal it to job applicants.}

In such settings, the decision maker may still have a vested interest in providing \emph{some} information about the decision rule to decision subjects in order to provide a certain level of \emph{transparency} and \emph{recourse}. In particular, the decision maker may be legally obliged, or economically motivated, to guide decision subjects to take actions that improve their underlying qualifications. To this end, instead of fully revealing the assessment rule, the decision maker can \emph{recommend actions} for decision subjects to take. Of course, such recommendations need to be chosen carefully and credibly; otherwise, self-interested decision subjects may not follow them or may utilize the recommendations to find pathways for manipulation.\looseness-1

We study a model of strategic learning in which the underlying assessment rule is not revealed to decision subjects. Our model captures several key aspects of the setting described above: First, even though the assessment rule is not revealed to the decision subjects, they often have \emph{prior knowledge} about what the rule may be. Secondly, when the decision maker provides recommendations to decision subjects on which action to take, the recommendations should be \emph{compatible with the subjects' incentives} to ensure they will follow the recommendation. Finally, our model assumes the decision maker discloses how they generate recommendations for recourse---an increasingly relevant requirement under recent regulations  (e.g., ~\citep{gdpr}).\looseness-1

Utilizing our model, we aim to design a mechanism for a decision maker to provide recourse to a decision subject who has incomplete information about the underlying assessment rule. We assume the assessment rule makes \emph{predictions} about some future outcome of the decision subject (e.g., whether they will pay back a loan in time if granted one). Before the assessment rule is trained (i.e., before the model parameters are fit), the decision maker and decision subject have some \emph{prior belief} about the realization of the assessment rule. This prior represents the ``common knowledge'' about the importance of various observable features for making accurate predictions. After training, the assessment rule is revealed to the decision maker, who then recommends an action for the decision subject to take, based on their pre-determined \emph{signaling policy}. Upon receiving this action recommendation, the decision subject updates their belief about the underlying assessment rule. They then take the action which they believe will maximize their utility (i.e., the benefit from the decision they receive, minus the cost of taking their selected action) in expectation. Finally, the decision maker uses the assessment rule to make a prediction about the decision subject.\looseness-1

The interaction described above is an instance of \emph{Bayesian persuasion}, a game-theoretic model of information revelation originally due to \citeauthor{kamenica2011bayesian}. The specific instance of Bayesian persuasion we consider in this work is summarized in Figure \ref{fig:summary}.\looseness-1

\begin{figure}
    \centering
    \noindent \fbox{\parbox{\textwidth}{
    \vspace{0.1cm}
    \textbf{Interaction protocol between the decision maker and decision subject} 
    
    \vspace{-0.2cm}
    \noindent\rule{13.9cm}{0.4pt}
    \begin{enumerate}[itemsep=0pt, parsep=1pt, leftmargin=8mm, rightmargin=5mm, topsep=1mm]
        \item The decision maker and decision subject initially have some prior/belief about the assessment rule that will be trained. 
        \item Before training, the decision maker commits to a signaling policy. After training, the assessment rule is revealed to the decision maker.
        \item The decision maker then uses their signaling policy and knowledge of the assessment rule to recommend an action for the decision subject to take.
        \item The decision subject updates their belief given the recommendation, and chooses an action that they believe maximizes their utility.
        \item The decision subject receives a prediction through the assessment rule.
    \end{enumerate}
    }}
    \caption{Summary of the setting we consider.}
    \label{fig:summary}
\end{figure}

\paragraph{Our contributions.} Our central conceptual contribution is to cast the problem of offering recourse under partial transparency as a game of Bayesian persuasion. Our key technical contributions consist of comparing optimal action recommendation policies in this new setup with two natural alternatives: (1) fully revealing the assessment rule to the decision subjects, or (2) revealing no information at all about the assessment rule. We provide new insights about the potentially significant advantages of action recommendation over these baselines, and offer efficient formulations to derive the optimal recommendations. More specifically, our analysis offers the following takeaways:\looseness-1
\begin{enumerate} [itemsep=0pt, parsep=2pt, leftmargin=8mm, topsep=0mm]
    \item Using tools from Bayesian persuasion, we show that it is possible for the decision maker to provide incentive-compatible action recommendations that encourage rational decision subjects to modify their features through beneficial interventions.
     While the decision maker and decision subjects are never worse off in expectation from using optimal incentive-compatible recommendations, we show that situations exist in which the decision maker is \emph{significantly better off} in expectation utilizing the optimal signaling policy (as opposed to the two baselines, Section~\ref{sec:ex}).\looseness=-1
    \item We derive the optimal signaling policy for the decision maker. While the decision maker's optimal signaling policy initially appears challenging to compute (as it involves optimizing over \emph{continuously-many} variables), we show that the problem can naturally be cast as a linear program defined in terms of a finite set of variables. However, solving this linear program may require reasoning about exponentially-many variables. Motivated by this observation, we provide a polynomial-time algorithm to approximate the optimal signaling policy (Section~\ref{sec:opt}).\looseness-1
    \item We empirically evaluate our persuasion mechanism on semi-synthetic data based on the Home Equity Line of Credit (HELOC) dataset, and find that the optimal signaling policy performs significantly better than the two natural alternatives across a wide range of instances (Section~\ref{sec:experiments}).
\end{enumerate}


%% file: related.tex
\subsection{Related work}\label{sec:related}
\noindent \xhdr{Strategic responses to unknown predictive models}
To the best of our knowledge, our work is the first to use tools from persuasion to model the strategic interaction between a decision maker and strategic decision subjects when the underlying predictive model is not public knowledge. Several prior works have addressed similar problems through different models and techniques. For example, \citet{akyol2016price} quantify the ``price of transparency'', a quantity which compares the decision maker's utility when the predictive model is fully known with their utility when the model is not revealed to the decision subjects. 
\citet{DBLP:conf/nips/TsirtsisR20} study the effects of counterfactual explanations on strategic behavior.
\citet{ghalme2021strategic} compare the prediction error of a classifier when it is public knowledge with the error when decision subjects must learn a version of it, and label this difference the ``price of opacity''. They show that small errors in decision subjects' estimates of the true underlying model may lead to large errors in the performance of the model. The authors argue that their work provides formal incentives for decision makers to adopt full transparency as a policy. Our work, in contrast, is based on the observation that even if decision makers are willing to reveal their models, legal requirements, privacy concerns, and intellectual property restrictions may prohibit full transparency. So we instead study the consequences of partial transparency---a common condition in real-world domains.\looseness-1

\citet{bechavod2021information} study the effects of information discrepancy across different sub-populations of decision subjects on their ability to improve their observable features in strategic learning settings. Like us, they do not assume the predictive model is fully known to the decision subjects. Instead, the authors model decision subjects as trying to infer the underlying predictive model by learning from their social circle of family and friends, which naturally causes different groups to form within the population. In contrast to this line of work, we study a setting in which the decision maker provides customized feedback to each decision subject individually. 
Additionally, while the models proposed by \citep{ghalme2021strategic,bechavod2021information} circumvent the assumption of full information about the deployed model, they restrict the decision subjects' knowledge to be obtained only through past data. 


\noindent \xhdr{Algorithmic recourse}
Our work is closely related to recent work on algorithmic recourse \cite{karimi2021survey}. Algorithmic recourse is concerned with providing explanations and recommendations to individuals who have received unfavorable automated decisions. A line of algorithmic recourse methods including \cite{wachter2017counterfactual, ustun2019actionable, joshi2019towards} focus on suggesting \emph{actionable} or realistic changes to underlying qualifications to decision subjects interested in improving their decisions. Our action recommendations are ``actionable'' in the sense that they are interventions which promote long-term desirable behaviors while ensuring that the decision subject is not worse off in expectation. 


\noindent \xhdr{Transparency} Recent legal and regulatory frameworks, such as the General Data Protection Regulation (GDPR) \cite{gdpr}, motivate the development of forms of algorithmic transparency suitable for real-world deployment. While this work can be thought of as providing additional transparency into the decision-making process, it does not naturally fall into the existing organizations of explanation methods (e.g., as outlined in \cite{chen2021towards}), as our policy does not simply recommend actions based on the decision rule. Rather, our goal is to incentivize actionable interventions on the decision subjects' observable features which are desirable to the decision maker, and we leverage persuasion techniques to ensure compliance.\looseness-1
%

\noindent \xhdr{Bayesian persuasion}
%
There has been growing interest in Bayesian persuasion \cite{kamenica2011bayesian} in the computer science and machine learning communities in recent years. \citet{dughmi2017algorithmic2, dughmi2019algorithmic} characterize the computational complexity of computing the optimal signaling policy for several popular models of persuasion. \citet{castiglioni2020online} study the problem of learning the receiver's utilities through repeated interactions. Work in the multi-arm bandit literature \cite{mansour2015bayesian, MansourSSW16, immorlica2019bayesian, chen2018incentivizing, sellke2021price} leverages techniques from Bayesian persuasion to incentivize agents to perform bandit exploration. Finally, linear programming-based approaches to Bayesian persuasion have been studied in the economics literature \cite{kolotilin2018optimal, dworczak2019simple}, although the persuasion setting we study differs considerably.

\noindent \xhdr{Other strategic learning settings}
The strategic learning literature \cite{hardt2016strategic, ghalme2021strategic, chen2021strategic, levanon2021strategic, jagadeesan2021alternative, bechavod2021information, harris2021stateful, harris2021strategic, kleinberg2020classifiers, frankel2019improving} broadly studies machine learning questions in the presence of strategic decision subjects. There has been a long line of work in strategic learning that focuses on how strategic decision subjects adapt their input to a machine learning algorithm in order to receive a more desirable prediction, although most prior work in this literature assumes that the underlying assessment rule is fully revealed to the decision subjects, which is typically not true in reality.

%% file: prelims.tex
\section{Setting and background}\label{sec:background}

Consider a setting in which a decision maker assigns a predicted label $\hat{y} \in \{-1, +1\}$ (e.g., whether or not someone will repay a loan if granted one) to a decision subject with \emph{initial} observable features $\vx_0 = (x_{0,1}, \cdots, x_{0,d-1}, 1)\in \mathbb{R}^d$ (e.g., amount of current debt, bank account balance, etc.).\footnote{We append a $1$ to the decision subject's feature vector for notational convenience.} We assume the decision maker uses a fixed linear decision rule to make predictions, i.e., $\hat{y} = \text{sign}\{\vx_0^\top \vtheta\}$, where the assessment rule $\vtheta \in \vTheta \subseteq \mathbb{R}^d$. 
The goal of the decision subject is to receive a positive classification (e.g., get approved for a loan). Given this goal, the decision subject may choose to take some \emph{action} $a$ from some set of possible actions $\mathcal{A}$ to modify their observable features (for example, they may decide to pay off a certain amount of existing debt, or redistribute their debt to game the credit score). We assume that the decision subject has $m$ actions $\{a_1, a_2, \ldots a_m\} \in \mathcal{A}$ at their disposal in order to improve their outcomes. For convenience, we add $a_{\emptyset}$ to $\mathcal{A}$ to denote taking "no action".  By taking action $a$, the decision subject incurs some \emph{cost} $c(a) \in \mathbb{R}$. This could be an actual monetary cost, but it can also represent non-monetary notions of cost such as opportunity cost or the time/effort the decision subject may have to exert to take the action. We assume taking an action $a$ changes a decision subject's observable feature values from $\vx_0$ to $\vx_0 + \Delta \vx(a)$, where $\Delta \vx(a) \in \mathbb{R}^d$, and $\Delta \vx_j(a)$ specifies the change in the $j$th observable feature as the result of taking action $a$.\footnote{Since we focus on a single decision subject, we hide the dependence of $\Delta \vx$ on the initial feature value $\vx_0$ to keep the notation simple.} For the special case of $a_{\emptyset}$, we have $\Delta \vx(a_{\emptyset}) = \mathbf{0}$, $c(a_{\emptyset}) = 0$. 
As a result of taking action $a$, a decision subject, \emph{ds}, receives utility $u_{\text{\text{ds}}}(a, \vtheta) = \text{sign}\{(\vx_0 + \Delta \vx(a))^\top \vtheta\} - c(a)$. In other words, the decision subject receives some positive (negative) utility for a positive (negative) classification, subject to a \emph{cost} for taking the action.

If the decision subject had exact knowledge of the assessment rule $\vtheta$ used by the decision maker, they could solve an optimization problem to determine the best action to take in order to maximize their utility. However, in many settings it is not realistic for a decision subject to have perfect knowledge of $\vtheta$. Instead, we model the decision subject's information through a \emph{prior} $\Pi$ over $\vtheta$, which can be thought of as ``common knowledge'' about the relative importance of various observable features in predicting the outcome of interest. For example, the decision subject may believe that prior payment history would likely be highly correlated with future default. We will use $\pi(\cdot)$ to denote the probability density function of $\Pi$ (so that $\pi(\vtheta)$ denotes the probability of the deployed assessment rule being $\vtheta$). We assume the decision subject is rational and risk-neutral. So at any point during the interaction, if they hold a belief $\Pi'$ about the underlying assessment rule, they would pick an action $a^*$ that maximize their \emph{expected} utility with respect to that belief. More precisely, they solve
$a^* \in \arg \max_{a \in \mathcal{A}} \mathbb{E}_{\vtheta \sim \Pi'}[u_{\text{\text{ds}}}(a, \vtheta)].$\looseness-1

From the decision maker's perspective, some actions may be more desirable than others. For example, a bank may prefer that an applicant pay off more existing debt than less when applying for a loan. To formalize this notion of action preference, we say that the decision maker receives some utility $u_{dm}(a) \in \mathbb{R}$ when the decision subject takes action $a$. In the loan example, $u_{dm}(\text{pay off more debt}) > u_{dm}(\text{pay off less debt})$. 

\subsection{Bayesian persuasion in the algorithmic recourse setting}\label{sec:bp}
The decision maker has an \emph{information advantage} over the decision subject, due to the fact that they know the true assessment rule $\vtheta$, whereas the decision subject does not. The decision maker may be able to leverage this information advantage to incentivize the decision subject to take a more favorable action (compared to the one they would have taken according to their prior) by recommending an action to the decision subject according to a commonly known \emph{signaling policy}. 

\begin{definition}[Signaling policy]
 A signaling policy $\mathcal{S}: \vTheta \rightarrow \mathcal{A}$ is a (possibly stochastic) mapping from assessment rules to actions.\footnote{Note that since our model is focused on the decision maker's interactions with a single decision subject, we drop the dependence of $\sigma$ on the decision subject's characteristics.} 
\end{definition}

We use $\sigma \sim \mathcal{S}(\vtheta)$ to denote the action recommendation sampled from signaling policy $\mathcal{S}$, where $\sigma \in \mathcal{A}$ is the realized recommended action. 

The decision maker's signaling policy is assumed to be fixed and common knowledge. This is because in order for the decision subject to perform a Bayesian update based on the observed recommendation, they need to know the signaling policy. Additionally, the decision maker must have the \emph{power of commitment}, i.e., the decision subject must believe that the decision maker
will select actions according to their signaling policy. In our setting, this will be the case since the decision maker commits to their signaling policy before training the assessment rule. This can be seen as a form of transparency, as the decision maker is publicly announcing how they will use their assessment rule to provide action recommendations/recourse before they train the assessment rule. For simplicity, we assume that the decision maker shares the same prior beliefs $\Pi$ as the decision subject over the observable features before the model is trained. These assumptions are standard in the Bayesian persuasion literature (see, e.g., \cite{kamenica2011bayesian, mansour2015bayesian, MansourSSW16}).

In order for the decision subject to be incentivized to follow the actions recommended by the decision maker, the signaling policy $\mathcal{S}$ needs to be \emph{Bayesian incentive-compatible}.

\begin{definition}[Bayesian incentive-compatibility]\label{def:BIC}
Consider a decision subject \emph{ds} with initial observable features $\vx_0$ and prior $\Pi$.
A signaling policy $\mathcal{S}$ is Bayesian incentive-compatible (BIC) for \emph{ds} if $    \mathbb{E}_{\vtheta \sim \Pi}[u_{ds}(a, \vtheta)|\sigma = a] \geq \mathbb{E}_{\vtheta \sim \Pi}[u_{ds}(a', \vtheta)|\sigma = a]$
 for all actions $a, a' \in \mathcal{A}$ such that $\mathcal{S}(\vtheta)$ has positive support on $\sigma = a$.
\end{definition}

In other words, a signaling policy $\mathcal{S}$ is BIC if, \emph{given that the decision maker recommends action} $a$, the decision subject's expected utility is at least as high as the expected utility of taking any other action $a'$.\looseness-1 

We remark that while for the ease of exposition our model focuses the interactions between the decision maker and a \emph{single} decision subject, our results can be extended to a heterogeneous population of decision subjects---as long as we assume their interactions with the decision-maker are independent of one another (e.g., this assumption rules out one subject updating their belief based on the outcome of another subject's prior interaction with the decision-maker). Under such a setting, the decision maker would publicly commit to a method of computing the signaling policy, given a decision subject's initial observable features as input. Once a decision subject arrives, their feature values are observed and the signaling policy is computed.

%% file: model-short.tex
\section{Characterizing the utility gains of persuasion} \label{sec:ex}
 
 As is generally the case in the persuasion literature \cite{kamenica2011bayesian, kamenica2019bayesian, dughmi2019algorithmic}, the decision maker can achieve higher expected utility with an optimized signaling policy compared to if they provided no recommendation or fully disclosed the model. To characterize \emph{how much} leveraging the decision maker's information advantage may improve their expected utility under our setting, we study the following example.\looseness-1

Consider a simple setting under which a single decision subject has one observable feature $x_0$ (e.g., credit score) and two possible actions: $a_{\emptyset} =$ ``do nothing'' (i.e., $\Delta x(a_{\emptyset}) = 0$, $c(a_{\emptyset}) = 0$, $u_{dm}(a_{\emptyset}) = 0$) and $a_1 =$ ``pay off existing debt'' (i.e., $\Delta x(a_{1}) > 0$, $c(a_{1}) > 0$, $u_{dm}(a_{1}) = 1$), which in turn raises their credit score. For the sake of our illustration, we assume credit-worthiness to be a mutually desirable trait, and credit scores to be a good measure of credit-worthiness. We assume the decision maker would like to design a signaling policy to maximize the chance of the decision subject taking action $a_1$, regardless of whether or not the applicant will receive the loan. In this simple setting, the decision maker's decision rule can be characterized by a single threshold parameter $\theta$, i.e., the decision subject receives a positive classification if $x + \theta \geq 0$ and a negative classification otherwise. Note that while the decision subject does not know the exact value of $\theta$, they instead have some prior over it, denoted by $\Pi$.\looseness-1

Given the true value of $\theta$, the decision maker recommends an action $\sigma \in \{a_\emptyset, a_1\}$ for the decision subject to take. The decision subject then takes a possibly different action $a \in \{a_\emptyset, a_1\}$, which changes their observable feature from $x_{0}$ to $x = x_{0} + \Delta x(a)$. Recall that the decision subject's utility takes the form $u_{ds}(a, \theta) = \text{sign}\{(x_{0} + \Delta x(a)) + \theta\} - c(a)$. Note that if $c(a_1) > 2$, then $u_{ds}(a_\emptyset, \theta) > u_{ds}(a_1, \theta)$ holds for any value of $\theta$, meaning that it is impossible to incentivize any rational decision subject to play action $a_1$. Therefore, in order to enable the decision maker to incentivize action $a_1$, we assume $c(a_1) < 2$. 

We observe that in this simple setting, we can bin values of $\theta$ into three different ``regions'', based on the outcome the decision subject would receive if $\theta$ were actually in that region. First, if $x_0 + \Delta x(a_1) + \theta < 0$, the decision subject will not receive a positive classification, even if they take action $a_1$. In this region, the decision subject's initial feature value $x_0$ is ``too low'' for taking the desired action to make a difference in their classification. We refer to this region as $L$. Second, if $x_0 + \theta \geq 0$, the decision subject will receive a positive classification \emph{no matter what} action they take. In this region, $x_0$ is ``too high'' for the action they take to make any difference on their classification. We refer to this region as $H$. Third, if $x_0 + \theta < 0$ and $x_0 + \Delta x(a_1) + \theta \geq 0$, the decision subject will receive a positive classification if they take action $a_1$ and a negative classification if they take action $a_\emptyset$. We refer to this region as $M$. Consider the signaling policy in Figure \ref{fig:sp}.

\begin{figure}
    \centering
    \noindent \fbox{\parbox{\textwidth}{
    \vspace{0.1cm}
    \textbf{Example signaling policy $\mathcal{S}(\theta)$}
    
    \vspace{-0.2cm}
    \noindent\rule{13.9cm}{0.4pt}
    \begin{itemize}[leftmargin=0.65in]
        \item[\textbf{Case 1:}] $\theta \in L$. Recommend action $a_1$ with probability $q$ and action $a_\emptyset$ with probability $1 - q$
        \item[\textbf{Case 2:}] $\theta \in M$. Recommend action $a_1$ with probability $1$
        \item[\textbf{Case 3:}] $\theta \in H$. Recommend action $a_1$ with probability $q$ and action $a_\emptyset$ with probability $1 - q$   
    \end{itemize}
    }}
    \caption{Signaling policy for the example of Section \ref{sec:ex}.}
    \label{fig:sp}
\end{figure}

In Case 2, $\mathcal{S}$ recommends the action ($a_1$) that the decision subject would have taken had they known the true $\theta$, with probability $1$. However, in Case 1 and Case 3, the decision maker recommends, with probability $q$, an action ($a_1$) that the decision subject would not have taken knowing $\theta$, leveraging the fact that the decision subject does not know exactly which case they are currently in. If the decision subject follows the decision maker's recommendation from $\mathcal{S}$, then the decision maker expected utility will increase from $0$ to $q$ if the realized $\theta \in L$ or $\theta \in H$, and will remain the same otherwise. Intuitively, if $q$ is ``small enough'' (where the precise definition of ``small'' depends on the prior over $\theta$ and the cost of taking action $a_1$), then it will be in the decision subject's best interest to follow the decision maker's recommendation, \textit{even though they know that the decision maker may sometimes recommend taking action $a_1$ when it is not in their best interest to take that action}. That is, the decision maker may recommend that a decision subject pay off existing debt with probability $q$ when it is unnecessary for them to do so in order to secure a loan. We now give a criteria on $q$ which ensures the signaling policy $\mathcal{S}$ is BIC.\looseness-1

\begin{prop}\label{thm:BIC-ex}
Signaling policy $\mathcal{S}$ is Bayesian incentive-compatible if $q = \min \{\frac{\pi(M)(2 - c(a_1))}{c(a_1)(1 - \pi(M))}, 1\}$, where $\pi(M) := \mathbb{P}_{\theta \sim \Pi}(x_{0} + \theta < 0 \text{ and } x_{0} + \Delta x(a_1) + \theta \geq 0)$.\looseness-1
\end{prop}

\textit{Proof Sketch.} We show that $\mathbb{E}_{\theta \sim \Pi}[u_{ds}(a_\emptyset, \theta)|\sigma = a_\emptyset] \geq \mathbb{E}_{\theta \sim \Pi}[u_{ds}(a_1, \theta)|\sigma = a_\emptyset]$ and $\mathbb{E}_{\theta \sim \Pi}[u_{ds}(a_1, \theta)|\sigma = a_1] \geq \mathbb{E}_{\theta \sim \Pi}[u_{ds}(a_\emptyset, \theta)|\sigma = a_1]$. Since these conditions are satisfied, $\mathcal{S}$ is BIC. The full proof may be found in Appendix \ref{sec:bic-ex-proof}.

Next, we show that the decision maker's expected utility when recommending actions according to the optimal signaling policy can be \emph{arbitrarily higher} than their expected utility from revealing full information or no information. We prove the following result in Appendix \ref{sec:unbounded-proof}.

\begin{prop}\label{prop:unbounded}
For any $\epsilon > 0$, there exists a problem instance such that the expected decision maker utility from recommending actions according to the optimal signaling policy is $1 - \epsilon$ and the expected decision maker utility for revealing full information or revealing no information is at most $\epsilon$.
\end{prop}

\section{Computing the optimal signaling policy}\label{sec:opt}


In Section~\ref{sec:ex}, we show a one-dimensional example, where a signaling policy can obtain arbitrarily better utilities compared to revealing full information and revealing no information. We now derive the decision maker's optimal signaling policy for the general setting with arbitrary numbers of observable features and actions described in Section \ref{sec:background}. Under the general setting, the decision maker's optimal signaling policy can be described by the following optimization:
\small
\begin{equation}\label{eqn:opt}
\begin{aligned}
    &\max_{p(\sigma = a|\vtheta), \forall a \in \mathcal{A}} \quad \mathbb{E}_{\sigma \sim \mathcal{S}(\vtheta), \vtheta \sim \Pi}[u_{dm}(\sigma)]\\
    \text{s.t.} \quad &\mathbb{E}_{\vtheta \sim \Pi}[u_{ds}(a, \vtheta) - u_{ds}(a', \vtheta) | \sigma = a] \geq 0, \; \forall a, a' \in \mathcal{A},
\end{aligned}
\end{equation}
\normalsize
where we omit the valid probability constraints over $p(\sigma = a|\vtheta), a \in \mathcal{A}$ for brevity. In words, the decision maker wants to design a signaling policy $\mathcal{S}$ in order to maximize their expected utility, subject to the constraint that the signaling policy is BIC. At first glance, the optimization may initially seem hopeless as there are infinitely many values of $p(\sigma = a|\vtheta)$ (one for every possible $\vtheta \in \vTheta$) that the decision maker's optimal policy must optimize over. However, we will show that the decision maker's optimal policy can actually be recovered by optimizing over finitely many variables. 

By rewriting the BIC constraints as integrals over $\vTheta$ and applying Bayes' rule, our optimization over $p(\sigma = a|\vtheta), a \in \mathcal{A}$ takes the following form
\small
\begin{equation*}
\begin{aligned}
    &\max_{p(\sigma = a|\vtheta), \forall a \in \mathcal{A}} \quad \mathbb{E}_{\sigma \sim \mathcal{S}(\vtheta), \vtheta \sim \Pi}[u_{dm}(\sigma)]\\
    \text{s.t.} \quad &\int_{\vTheta} p(\sigma = a|\vtheta)\pi(\vtheta)(u_{ds}(a, \vtheta) - u_{ds}(a', \vtheta)) d\vtheta \geq 0, \; \forall a, a' \in \mathcal{A}.
\end{aligned}
\end{equation*}
\normalsize

Note that if $u_{ds}(a, \vtheta) - u_{ds}(a', \vtheta)$ is the same for some ``equivalence region'' $R \subseteq \vTheta$ (which we formally define below), we can pull $u_{ds}(a, \vtheta) - u_{ds}(a', \vtheta)$ out of the integral and instead sum over the different equivalence regions. Intuitively, an equivalence region can be thought of as the set of all $\vtheta \in \vTheta$ pairs that are indistinguishable from a decision subject's perspective because they lead to the exact same utility for any possible action the decision subject could take.
Based on this idea, we formally define a region of $\vTheta$ as follows. 

\begin{definition}[Equivalence Region]\label{def:region}
Two assessments $\vtheta, \vtheta'$ are \emph{equivalent} (w.r.t. $u_{ds}$) if $u_{ds}(a, \vtheta) - u_{ds}(a', \vtheta) = u_{ds}(a, \vtheta') - u_{ds}(a', \vtheta')$, $\forall a, a' \in \mathcal{A}$. An \emph{equivalence region} $R$ is a subset of $\vTheta$ such that for any $\theta \in R$, all $\theta'$ equivalent to $\theta$ are also in $R$. We denote the set of all equivalence regions by $\mathcal{R}$.\looseness-1
\end{definition}

For more intuition about the definition of an equivalence region, see Figure \ref{fig:region} in Appendix \ref{sec:region-illustration}. After pulling the decision subject utility function out of the integral, we can integrate $p(\sigma = a|\vtheta) \pi(\vtheta)$ over each equivalence region $R$. We denote $p(R)$ as the probability that the true $\vtheta \in R$ according to the prior. Finally, since it is possible to write the constraints in terms of $p(\sigma = a|R)$, $\forall a \in \mathcal{A}, R \in \mathcal{R}$, it suffices to optimize directly over these quantities. For completeness, we include the constraints which make each $\{p(\sigma = a|R)\}_{a \in \mathcal{A}}$, $\forall R$ a valid probability distribution.

\begin{theorem}[Optimal signaling policy]\label{thm:opt}
The decision maker's optimal signaling policy can be characterized by the following linear program OPT-LP:
\small
\begin{equation}\label{opt:exact}
\tag*{(OPT-LP)}
\begin{aligned}
    \max_{p(\sigma = a|R), \forall a \in \mathcal{A}, R \in \mathcal{R}} \quad &\sum_{a \in \mathcal{A}} \sum_{R \in \mathcal{R}} p(R) p(\sigma = a|R) u_{dm}(a)\\
    \text{s.t.} \quad &\sum_{R \in \mathcal{R}} p(\sigma = a|R)p(R)(u_{ds}(a, R) - u_{ds}(a', R)) \geq 0, \; \forall a, a' \in \mathcal{A}\\
    &\sum_{a \in \mathcal{A}} p(\sigma = a|R) = 1, \; \forall R, \; \; \;
    p(\sigma = a|R) \geq 0, \; \forall R \in \mathcal{R}, a \in \mathcal{A},\\
\end{aligned}
\end{equation}
\normalsize
\end{theorem}

where $p(\sigma = a|R)$ denotes the probability of sending recommendation $\sigma = a$ if $\vtheta \in R$.
Note that the linear program OPT-LP is always feasible, as the decision maker can always recommend the action the decision subject would play according to the prior, which is BIC. Similarly, always recommending the action the decision subject would take had they known the assessment rule is also feasible.

While the problem of determining the decision maker's optimal signaling policy can be transformed from an optimization over infinitely many variables into an optimization over the set of finitely many equivalence regions $\mathcal{R}$, $|\mathcal{R}|$ may be exponential in the number of observable features $d$ (see Appendix \ref{sec:regions} for more details). This is perhaps unsurprising as without any assumptions on $\vTheta$, the representation of the prior $\Pi$ can scale exponentially with $d$. In this case, we expect the running time of an algorithm which takes the entire prior as input to be exponential in the number of features as well. This motivates the need for a computationally efficient algorithm to approximate \ref{opt:exact}, which does not require the full prior $\Pi$ as input.\looseness-1

We adapt the sampling-based approximation algorithm of \citet{dughmi2019algorithmic} to our setting in order to compute an $\epsilon$-optimal and $\epsilon$-approximate signaling policy in polynomial time, as shown in Algorithm \ref{alg:approx} in Appendix \ref{sec:approx_appendix}. At a high level, Algorithm \ref{alg:approx} samples polynomially-many times from the prior distribution over the space of assessment rules, and solves an empirical analogue of \ref{opt:exact}. In Appendix \ref{sec:approx_appendix}, we show that the resulting signaling policy is $\epsilon$-BIC, and is $\epsilon$-optimal with high probability, for any $\epsilon > 0$. Formally, we prove the following statement.

\begin{theorem}\label{thm:approx}
Algorithm \ref{alg:approx} runs in poly($m, \frac{1}{\epsilon}$) time (where $m = |\mathcal{A}|$), and implements an $\epsilon$-BIC signaling policy that is $\epsilon$-optimal with probability at least $1 - \delta$.
\end{theorem}

We leave open the question of whether there are classes of succinctly represented priors that permit efficient algorithms for computing the exact optimal policy in time polynomial in $d$ and $m$. 
It is also plausible to design efficient algorithms that only require some form of query access to the prior distribution. However, information-theoretic lower bounds of \cite{dughmi2019algorithmic} rule out query access through sampling.\looseness-1

%% file: experiments.tex
\section{Experiments}\label{sec:experiments}

In this section, we provide experimental results using a semi-synthetic setting where decision subjects are based on individuals in the Home Equity Line of Credit (HELOC) dataset \cite{FICO}. The HELOC dataset contains information about 9,282 customers who received a Home Equity Line of Credit. Each individual in the dataset has 23 observable features related to an applicant's financial history (e.g., percentage of previous payments that were delinquent) and a label which characterizes their loan repayment status (repaid/defaulted). We compare the decision maker utility for different models of information revelation: our optimal signaling policy, revealing full information about the model, revealing no information about the model. As our theory predicts, the expected decision maker utility when recommending actions according to the optimal signaling policy either matches or exceeds the expected utility from revealing full information or no information about the assessment rule across all problem instances. Moreover, the expected decision maker utility from signaling is \emph{significantly} higher on average. Next, we explore how the decision maker's expected utility changes when action costs and changes in observable features are varied jointly. Our results are summarized in Figures \ref{fig:util-var-comparison} and \ref{fig:3d-plot}.\looseness-1

In order to adapt the HELOC dataset to our strategic setting, we select four features and define five hypothetical actions $\mathcal{A} = \{a_\emptyset, a_1, a_2, a_3, a_4\}$ that decision subjects may take in order to improve their observable features. Actions $\{a_1, a_2, a_3, a_4\}$ result in changes to each of the decision subject's four observable features, whereas action $a_\emptyset$ does not. For simplicity, we view actions $\{a_1, a_2, a_3, a_4\}$ as equally desirable to the decision maker, and assume they are all more desirable than $a_\emptyset$. Using these four features, we train a logistic regression model that predicts whether an individual is likely to pay back a loan if given one, which will serve as the decision maker's realized assessment rule. For more information on how we constructed our experiments, see Appendix~\ref{sec:appdx-decision-maker-model}.

\begin{figure}
    \centering
    \floatbox[{\capbeside\thisfloatsetup{capbesideposition={right, center},capbesidewidth=0.4\textwidth}}]{figure}[\FBwidth]
    {\caption{Total decision maker utility 
    averaged across all cost and $\Delta \vx(a)$ configurations 
    for three different prior variances.
    The optimal signaling policy (red) consistently yields higher utility compared to the two baselines: revealing full information (blue) and no information (green). This gap increases when the decision subject is less certain about the model being used (higher $\sigma^2$).\looseness-1
    }}
    {\includegraphics[width=0.5\textwidth]{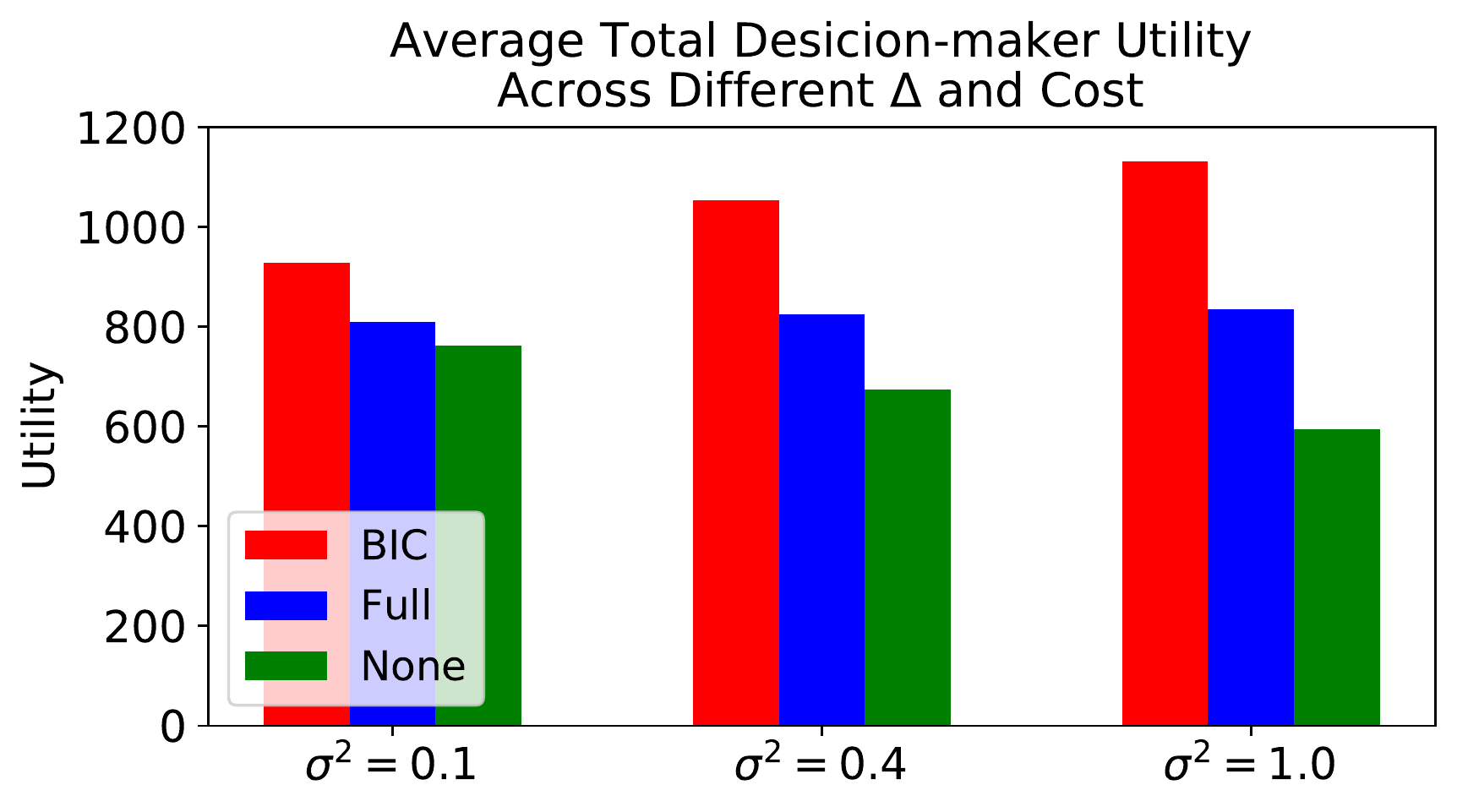}\label{fig:util-var-comparison}}
\end{figure}

Given a $\{(c(a_i), \Delta \vx(a_i))\}_{i=1}^4$ instance and information revelation scheme, we calculate the decision maker's total expected utility by summing their expected utility for each applicant. Figure~\ref{fig:util-var-comparison} shows the average total expected decision maker utility across different $\Delta \vx(a)$ and cost configurations for priors with varying amounts of uncertainty. See Figure \ref{fig:additional-details-configs} in Appendix \ref{sec:appdx-additional-results} for plots of all instances which were used to generate Figure \ref{fig:util-var-comparison}. Across all instances, the optimal signaling policy (red) achieves higher average total utility compared to the other information revelation schemes (blue and green). The difference is further amplified whenever the decision subjects are less certain about the true assessment rule (i.e., when $\sigma$ is large). Intuitively, this is because the decision maker leverages the decision subjects' uncertainty about the true assessment rule in order to incentivize them to take desirable actions, and as the uncertainty increases, so does their ability of persuasion.

\begin{figure}
    \centering
    \includegraphics[width=\textwidth]{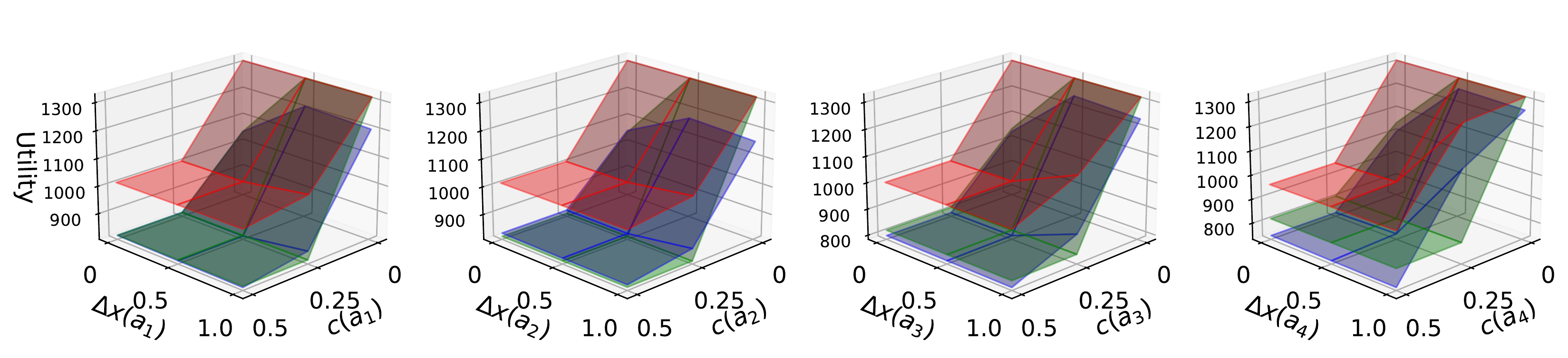}
    \caption{Expected utility across different $c(a)$ and $\Delta \vx(a)$ configurations for $\sigma^2=0.4$. Optimal signaling policy (red) effectively upper-bounds the two baselines, revealing everything (blue) and revealing nothing (green) in all settings.}
    \label{fig:3d-plot}
\end{figure}

To better understand how the decision maker's expected utility changes as a function of $c(a)$ and $\Delta \vx(a)$, we sweep through multiple $\{(c(a_i), \Delta \vx(a_i))\}_{i=1}^4$ tuples on a grid of $( c(a_i), \Delta \vx(a_i) ) \in \{0, 0.25, 0.5\} \times \{0, 0.5, 1.0\}$ for $i \in \{1,2,3,4\}$ and measure the effectiveness of the three information revelation schemes. Figure~\ref{fig:3d-plot} shows the surface of the decision maker utility as a function of $( c(a_i), \Delta \vx(a_i) )$ for the optimal signaling policy (red), revealing full information (blue), and revealing no information (green). When $c(a_i)$ is high and $\Delta \vx (a_i)$ is low, the total expected decision maker utility is low as there is less incentive for the decision subject to take actions (although even under this setting, the optimal signaling policy outperforms the other two baselines). As $c(a_i)$ decreases and $\Delta \vx (a_i)$ increases, the total expected decision maker utility increases.

%% file: conclusion.tex
\section{Conclusion}\label{sec:conc}
We investigate the problem of offering algorithmic recourse without requiring full transparency (i.e., revealing the assessment rule). We cast this problem as a game of Bayesian persuasion, and offer several new insights regarding how a decision maker can leverage their information advantage over decision subjects to incentivize mutually beneficial actions. Our stylized model relies on several simplifying assumptions, which suggest important directions for future work: 

\noindent \xhdr{Public persuasion} We assume that the recommendations received by each decision subject are \emph{private}. However, if a decision subject is given access to recommendations for multiple individuals, it may be possible for them to reconstruct the underlying model. While out of the scope of this work, it would be interesting to study models of \emph{public} persuasion in the algorithmic recourse setting.\looseness-1

\noindent \xhdr{Beyond linear decision rules} We focus on settings with \emph{linear} decision rules and assume all decision subject parameters (e.g., cost function, initial observable features, etc.) are known to the decision maker. We leave it for future work to extend our findings to non-linear decision rules, or settings in which some of the decision subjects' parameters are unknown to the decision maker. 

%% file: acks.tex
\section{Acknowledgements}

KH is supported by a NDSEG Fellowship. ZSW and KH were supported in part by the NSF FAI Award \#1939606, a Google Faculty Research Award, a J.P. Morgan Faculty Award, a Facebook Research Award, and a Mozilla Research Grant. AT was supported in part by the National Science Foundation grants IIS1705121, IIS1838017, IIS2046613, IIS2112471, a funding from Meta, Morgan Stanley and Amazon. HH acknowledges support from NSF IIS2040929, a CyLab 2021 grant, and a Meta (Facebook) research award. Any opinions, findings and conclusions or recommendations expressed in this material are those of the author(s) and do not necessarily reflect the views of any of these funding agencies. KH would like to thank Haifeng Xu for insightful conversations about \citet{dughmi2019algorithmic}. KH would also like to thank James Best, Yatong Chen, Jeremy Cohen, Daniel Ngo, Chara Podimata, and Logan Stapleton for helpful suggestions and conversations.

%% file: checklist.tex
\section*{Checklist}

\begin{enumerate}

\item For all authors...
\begin{enumerate}
  \item Do the main claims made in the abstract and introduction accurately reflect the paper's contributions and scope?
    \answerYes{}
  \item Did you describe the limitations of your work?
    \answerYes{} See Section \ref{sec:conc}.
  \item Did you discuss any potential negative societal impacts of your work?
    \answerYes{} See Appendix \ref{sec:si}. 
  \item Have you read the ethics review guidelines and ensured that your paper conforms to them?
    \answerYes{}
\end{enumerate}

\item If you are including theoretical results...
\begin{enumerate}
  \item Did you state the full set of assumptions of all theoretical results?
    \answerYes{} See Sections \ref{sec:ex} and \ref{sec:opt}.
        \item Did you include complete proofs of all theoretical results?
    \answerYes{} See Appendices \ref{sec:bic-ex-proof}, \ref{sec:unbounded-proof}, \ref{sec:regions}, \ref{sec:approx_appendix}.
\end{enumerate}

\item If you ran experiments...
\begin{enumerate}
  \item Did you include the code, data, and instructions needed to reproduce the main experimental results (either in the supplemental material or as a URL)?
    \answerYes{} See supplementary material.
  \item Did you specify all the training details (e.g., data splits, hyperparameters, how they were chosen)?
    \answerYes{} See Appendix \ref{sec:1d-experiment} and \ref{sec:appdx-decision-maker-model}.
        \item Did you report error bars (e.g., with respect to the random seed after running experiments multiple times)?
    \answerNA{}
        \item Did you include the total amount of compute and the type of resources used (e.g., type of GPUs, internal cluster, or cloud provider)?
    \answerYes{} See Appendix \ref{sec:compute}
\end{enumerate}

\item If you are using existing assets (e.g., code, data, models) or curating/releasing new assets...
\begin{enumerate}
  \item If your work uses existing assets, did you cite the creators?
    \answerNA{}
  \item Did you mention the license of the assets?
    \answerNA{}
  \item Did you include any new assets either in the supplemental material or as a URL?
    \answerNA{}
  \item Did you discuss whether and how consent was obtained from people whose data you're using/curating?
    \answerNA{}
  \item Did you discuss whether the data you are using/curating contains personally identifiable information or offensive content?
    \answerNA{}
\end{enumerate}

\item If you used crowdsourcing or conducted research with human subjects...
\begin{enumerate}
  \item Did you include the full text of instructions given to participants and screenshots, if applicable?
    \answerNA{}
  \item Did you describe any potential participant risks, with links to Institutional Review Board (IRB) approvals, if applicable?
    \answerNA{}
  \item Did you include the estimated hourly wage paid to participants and the total amount spent on participant compensation?
    \answerNA{}
\end{enumerate}

\end{enumerate}

%% file: social.tex
\section{Societal implications}\label{sec:si}
Persuasion in the algorithmic recourse setting is usually socially beneficial, as the quantities the decision maker wants to incentivize are usually aligned with socially desirable actions. Our work emphasizes these potential social benefits through our running example with credit lending which illustrates how persuasion can be used to incentivize individuals to pay off more existing debt. However, we recognize that there may be decision makers whose interests are negatively correlated with the interests of society who may use persuasion as a tool to incentivize socially undesirable behavior. Our work illustrates the ways in which persuasion may play a role in the deployment of algorithmic decision-making, capturing important dynamics between the decision maker and decision subject in the study of the long-term impact of these systems. 

%% file: exp-compute.tex
\section{Computational Resources}\label{sec:compute}
We ran our experiments on an Intel Core i7-8700 3.2 GHz 6-Core Processor with 16GB of RAM and a GTX1070 GPU, although the GPU was not used.

%% file: example_appendix.tex
\section{Proof of Proposition \ref{thm:BIC-ex}}\label{sec:bic-ex-proof}
\begin{proof}

Based on the decision subject's prior over $\theta$, they can calculate 
\begin{itemize}
    \item[(1)] $\pi(L) = \mathbb{P}_{\Pi}(x_{0} + \Delta x(a_1) + \theta < 0)$, i.e., the probability the decision subject is in region $L$ according to the prior
    \item[(2)] $\pi(M) = \mathbb{P}_{\Pi}(x_{0} + \theta < 0 \text{ and } x_{0} + \Delta x(a_1) + \theta \geq 0)$, i.e., the probability the decision subject is in region $M$ according to the prior 
    \item[(3)] $\pi(H) = \mathbb{P}_{\Pi}(x_{0} + \theta \geq 0)$, i.e., the probability the decision subject is in region $H$ according to the prior 
\end{itemize}

\textbf{Case 1}: $\sigma = a_0$. Given the signal $\sigma = a_\emptyset$, the decision subject's \emph{posterior} probability density function $\pi(\cdot|\sigma = a_\emptyset)$ over $L$, $M$, and $H$ will take the form 

\begin{itemize}
    \item[] $\pi(L|\sigma = a_\emptyset) = \frac{p(\sigma = a_\emptyset|L)\pi(L)}{p(\sigma = a_\emptyset)} = \frac{\pi(L)}{\pi(L) + \pi(H)}$ 
    \item[] $\pi(M|\sigma = a_\emptyset) = \frac{p(\sigma = a_\emptyset|M)\pi(M)}{p(\sigma = a_\emptyset)} = 0$
    \item[] $\pi(H|\sigma = a_\emptyset) = \frac{p(\sigma = a_\emptyset|H)\pi(H)}{p(\sigma = a_\emptyset)} = \frac{\pi(H)}{\pi(L) + \pi(H)}$
\end{itemize}

If the decision subject receives signal $\sigma = a_0$, they know with probability $1$ that they are \emph{not} in region $M$ with probability $1$. Therefore, they know that taking action $a_1$ will not change their classification, so they will follow the decision maker's recommendation and take action $a_\emptyset$.

\textbf{Case 2}: $\sigma = a_1$. Given the signal $\sigma = a_1$, the decision subject's posterior density over $L$, $M$, and $H$ will take the form 
\begin{itemize}
    \item[] $\pi(L|\sigma = a_1) = \frac{p(\sigma = a_1|L)\pi(L)}{p(\sigma = a_1)} = \frac{q \pi(L)}{\pi(M) + q (\pi(L) + \pi(H))} = \frac{q \pi(L)}{\pi(M) + q (1 - \pi(M))}$ 
    \item[] $\pi(M|\sigma = a_1) = \frac{p(\sigma = a_1|M)\pi(M)}{p(\sigma = a_1)} = \frac{\pi(M)}{\pi(M) + q (\pi(L) + \pi(H))} = \frac{\pi(M)}{\pi(M) + q (1 - \pi(M))}$
    \item[] $\pi(H|\sigma = a_1) = \frac{p(\sigma = a_1|H)\pi(H)}{p(\sigma = a_1)} = \frac{q \pi(H)}{\pi(M) + q (\pi(L) + \pi(H))} = \frac{q \pi(H)}{\pi(M) + q (1 - \pi(M))}$
\end{itemize}

\noindent The decision subject's expected utility of taking actions $a_\emptyset$ and $a_1$ under the posterior induced by $\sigma = a_1$ are
\begin{equation*}
\begin{aligned}
    \mathbb{E}_{\vtheta \sim \Pi}[u_{ds}(a_\emptyset, \theta)|\sigma = a_1] &= \pi(H|\sigma = a_1) \cdot (1 - 0) + \pi(M|\sigma = a_1) \cdot (-1 - 0) + \pi(L|\sigma = a_1) \cdot (-1 - 0)\\
    &= \pi(H|\sigma = a_1) - \pi(M|\sigma = a_1) - \pi(L|\sigma = a_1)
\end{aligned}
\end{equation*}
and
\begin{equation*}
    \begin{aligned}
        \mathbb{E}_{\vtheta \sim \Pi}[u_{ds}(a_1, \theta)|\sigma = a_1] &= \pi(H|\sigma = a_1) \cdot (1 - c(a_1))\\ &+ \pi(M|\sigma = a_1) \cdot (1 - c(a_1)) + \pi(L|\sigma = a_1) \cdot (-1 - c(a_1))\\
    \end{aligned}
\end{equation*}

\noindent In order for $\mathcal{S}$ to be BIC,
\begin{equation*}
    \mathbb{E}_{\vtheta \sim \Pi}[u_{ds}(a_1, \theta)|\sigma = a_1] \geq \mathbb{E}_{\vtheta \sim \Pi}[u_{ds}( a_\emptyset, \theta)|\sigma = a_1].
\end{equation*}
Plugging in our expressions for $\mathbb{E}_{\vtheta \sim \Pi}[u_{ds}(a_1, \theta)|\sigma = a_1]$ and $\mathbb{E}_{\vtheta \sim \Pi}[u_{ds}( a_\emptyset, \theta)|\sigma = a_1]$, we see that

\begin{equation*}
\begin{aligned}
    \pi(H|\sigma = a_1) \cdot (1 - c(a_1)) + &\pi(M|\sigma = a_1) \cdot (1 - c(a_1)) + \pi(L|\sigma = a_1) \cdot (-1 - c(a_1)) \\
    &\geq \pi(H|\sigma = a_1) - \pi(M|\sigma = a_1) - \pi(L|\sigma = a_1)
\end{aligned}
\end{equation*}

\noindent After canceling terms and simplifying, we see that
\begin{equation*}
    -(\pi(L|\sigma = a_1) + \pi(H|\sigma = a_1)) c(a_1) + \pi(M|\sigma = a_1) (2 - c(a_1)) \geq 0
\end{equation*}

\noindent Next, we plug in for $\pi(L|\sigma = a_1)$, $\pi(M|\sigma = a_1)$, and $\pi(H|\sigma = a_1)$. Note that the denominators of $\pi(L|\sigma = a_1)$, $\pi(M|\sigma = a_1)$, and $\pi(H|\sigma = a_1)$ cancel out.

\begin{equation*}
    -q(\pi(L) + \pi(H))c(a_1) + \pi(M) (2 - c(a_1)) = -q(1 - \pi(M))c(a_1) + \pi(M) (2 - c(a_1)) \geq 0
\end{equation*}

\noindent Solving for $q$, we see that
\begin{equation*}
    q \leq \frac{\pi(M)(2 - c(a_1))}{c(a_1)(1 - \pi(M))}.
\end{equation*}

\noindent Note that $q \geq 0$ always. Finally, in order for $q$ to be a valid probability, we restrict $q$ such that
\begin{equation*}
    q = \min\{ \frac{\pi(M)(2 - c(a_1))}{c(a_1)(1 - \pi(M))}, 1\}.
\end{equation*}
\noindent This completes the proof.
\end{proof}

\section{Proof of Proposition \ref{prop:unbounded}}\label{sec:unbounded-proof}
\begin{proof}
Consider the example in Section \ref{sec:ex}. 

\noindent \xhdr{Expected utility from revealing no information} If the decision subject acts exclusively according to the prior, they will select action $a_1$ with probability $1$ if $\mathbb{E}_{\vtheta \sim \Pi}[u_{ds}(a_1, \theta)] \geq \mathbb{E}_{\vtheta \sim \Pi}[u_{ds}(a_\emptyset, \theta)]$ and with probability $0$ otherwise. Plugging in our expressions for $\mathbb{E}_{\vtheta \sim \Pi}[u_{ds}(a_1, \theta)]$ and $\mathbb{E}_{\vtheta \sim \Pi}[u_{ds}(a_\emptyset, \theta)]$, we see that the decision subject will select action $a_1$ only if 
\begin{equation*}
    \pi(L)(-1 - c(a_1)) + \pi(M)(1 - c(a_1)) + \pi(H)(1 - c(a_1)) \geq \pi(L) (-1 - 0) + \pi(M)(-1 - 0) + \pi(H)(1 - 0)
\end{equation*}

\noindent Canceling terms and simplifying, we see that 
\begin{equation*}
    -c(a_1)(\pi(L) + \pi(H)) + \pi(M)(2 - c(a_1)) \geq 0
\end{equation*}
must hold for the decision subject to select action $a_1$. Finally, substituting $\pi(L) + \pi(H) = 1 - \pi(M)$ gives us the condition $2\pi(M) -c(a_1) \geq 0$. Alternatively, if $\frac{\pi(M)}{c(a_1)} < \frac{1}{2}$,
the decision subject will select action $a_\emptyset$ with probability $1$. Intuitively, this means that a rational decision subject would take action $a_1$ if the ratio of $\pi(M)$ (the probability according to the prior that taking action $a_1$ is in the decision subject's best interest) to $c(a_1)$ (the cost of taking action $a_1$) is high, and would take action $a_\emptyset$ otherwise.

\noindent \xhdr{Expected utility from revealing full information} If the decision maker reveals the assessment rule to the decision subject, they will select action $a_1$ when $\theta \in M$ and action $a_\emptyset$ otherwise. Therefore since $u_{dm}(a_1) = 1$ and $u_{dm}(a_\emptyset) = 0$, the decision maker's expected utility if they reveal full information is $\pi(M)$.

\noindent \xhdr{Expected utility from $\mathcal{S}$} Recall that the decision maker's signaling policy $\mathcal{S}$ from Section \ref{sec:ex} sets $q = \min\{\frac{\pi(M)(2 - c(a_1))}{c(a_1)(1 - \pi(M))}, 1\}$. Under this setting, the decision maker's expected utility is $\min\{1\cdot\pi(M) + q\cdot (1 - \pi(M)), 1\}$. Substituting in our expression for $q$ and simplifying, we see that the decision maker's expected utility for recommending actions via $\mathcal{S}$ is $\min \{\frac{2\pi(M)}{c(a_1)}, 1\}$.

Suppose that $2 \pi(M) = c(a_1) (1 - \epsilon)$ and $c(a_1) = 2\epsilon$, for some small $\epsilon > 0$. 
The decision maker's expected utility will always be $0$ from revealing no information because $\frac{2\pi(M)}{c(a_1)} = 1 - \epsilon < 1$. The decision maker's expected utility from recommending actions via $\mathcal{S}$ will be $\frac{2\pi(M)}{c(a_1)} = 1 - \epsilon$ . Since $\pi(M) = \epsilon(1 - \epsilon) < \epsilon$, the decision maker's expected utility from revealing full information will be less than $\epsilon$. Therefore, as $\epsilon$ approaches $0$, the decision maker's expected utility from revealing full information approaches $0$ (the smallest value possible), and the decision maker's expected utility from $\mathcal{S}$ approaches $1$ (the highest value possible). This completes the proof.
\end{proof}

The decision maker's expected utility as a function of their possible strategies is summarized in Table \ref{table:payoff}. Note that when $\mathbbm{1}\{\pi(M) \geq \frac{c(a_1)}{2}\} = 1$, $q = 1$. Therefore, the decision maker's expected utility is always as least as good as the two natural alternatives of revealing no information about the assessment rule, or revealing full information about the rule.

\begin{table}
    \centering
    \begin{tabular}{c c c c}
         & No information & Signaling with $\mathcal{S}$ & Full information \\ \hline
         Decision maker utility & $\mathbbm{1}\{\pi(M) \geq \frac{c(a_1)}{2}\}$ & $\pi(M) + q(1 - \pi(M))$ & $\pi(M)$ \\ \hline
    \end{tabular}
    \caption{Decision maker's expected utility when (1) revealing no information about the model, (2) recommending actions according to $\mathcal{S}$, and (3) revealing full information about the model.}\label{table:payoff}
\end{table}

%% file: computational_barriers.tex
\section{Region Illustration}\label{sec:region-illustration}
In Figure \ref{fig:region}, we show an example of how different equivalence regions might partition the space of possible assessment rules $\vTheta$. In this example, there are two actions and two observable features, and the space of $\vTheta$ is partitioned into three different equivalence regions. Note that as long as the set of actions $\mathcal{A}$ is finite, $|\mathcal{R}| < \infty$. 

\begin{figure}[t]
    \centering
      \begin{tikzpicture}[scale=1.5,extended line/.style={shorten >=-#1,shorten <=-#1},every text node part/.style={align=center}]
\draw[draw=black] (0,0) rectangle ++(1,1);
\draw[draw=black] (1,0) rectangle ++(3,1);
\draw[draw=black] (0,1) rectangle ++(1,1);
\draw[draw=black] (1,1) rectangle ++(3,1);
\draw [->](0,0)--(0,2.2) node[right]{$\theta_2$};
\draw [->](0,0)--(4.2,0) node[right]{$\theta_1$};
\foreach \x/\xtext in {0, 1/\frac{1}{2}, 4/2}
{\draw (\x cm,1pt ) -- (\x cm,-1pt ) node[anchor=north] {$\xtext$};}
\foreach \y/\ytext in {1/\frac{1}{2}, 2/1}
{\draw (1pt,\y cm) -- (-1pt ,\y cm) node[anchor=east] {$\ytext$};}
\node at (0.5,0.5) {$R_0$\\ $\{\emptyset\}$};
\node at (2.5,0.5) {$R_1$\\ $\{a_1\}$};
\node at (2.5,1.5) {$R_0$\\ $\{a_1,a_2\}$};
\node at (0.5,1.5) {$R_2$\\ $\{a_2\}$};
\end{tikzpicture}
\caption{An illustration of the equivalence regions for a two action ($a_1$, $a_2$) and two observable feature ($x_1, x_2$) setting, where $\vTheta = [0, 2] \times [0, 1] \times \{\frac{1}{2}\}$. Consider an individual with $\vx_0 = [0, 0, 1]^\top$, $\Delta \vx(a_1) = [1, 0, 0]^\top$, and $\Delta \vx(a_2) = [0, 1, 0]^\top$. The equivalence regions of $\vTheta$ are quadrants described the set of actions the decision subject could take in order to receive a positive classification. Region $R_0$ contains the bottom-left and top-right quadrants of $\vTheta$, region $R_1$ contains the bottom-right quadrant of $\vTheta$, and region $R_2$ contains the top-left quadrant of $\vTheta$.}
  \label{fig:region}
\end{figure}
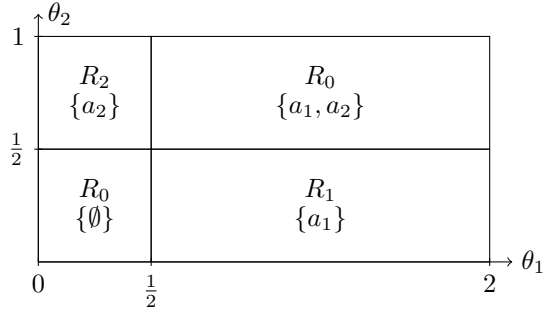

\section{Computational Barriers}\label{sec:regions}
In this section, we show that even in the setting where each action only affects one observable feature (e.g., as shown in Figure \ref{fig:actionscheme}), the number of equivalence regions in \ref{opt:exact} is still exponential in the size of the input. 
While somewhat simplistic, we believe this action scheme reasonably reflects real-world settings in which the decision subjects are under time or resource constraints when deciding which action to take. For example, the decision subject may need to choose between paying off some amount of debt and opening a new credit card when strategically modifying their observable features before applying for a loan.

\usetikzlibrary{positioning}
\tikzset{main node/.style={circle,draw,minimum size=1cm,inner sep=0pt, thick},}
\begin{figure}
    \centering
      \begin{tikzpicture}

    \begin{scope}[xshift=4cm]
    \node[main node] (1) {$a_{i,1}$};
    \node[main node] (2) [above = 0.5cm  of 1]  {$a_{d,1}$};
    \node[main node] (3) [left = 1cm  of 1] {$a_{\emptyset}$};
    \node[main node] (4) [below = 0.5cm  of 1] {$a_{1,1}$};
    \node[main node] (5) [right = 1.5cm  of 4] {...};
    \node[main node] (6) [right = 1.5cm  of 5] {$a_{1,m_1}$};
    \node[main node] (7) [right = 1.5cm  of 1] {...};
    \node[main node] (8) [right = 1.5cm  of 7] {$a_{i,m_i}$};
    \node[main node] (9) [right = 1.5cm  of 2] {...};
    \node[main node] (10) [right = 1.5cm  of 9] {$a_{d,m_d}$};
    \node at (1,1) {\rotatebox{90}{{...}}};
    \node at (1,-1) {\rotatebox{90}{{...}}};
    \path[draw,thick]
    (3) edge node {} (4)
    (3) edge node {} (1)
    (3) edge node {} (2)
    (4) edge node {} (5)
    (5) edge node {} (6)
    (1) edge node {} (7)
    (7) edge node {} (8)
    (2) edge node {} (9)
    (9) edge node {} (10)
    ;
    \end{scope}
\end{tikzpicture}
    \caption{Graphical representation of special ordering over the actions available to each decision subject. Each branch corresponds to an observable feature and each node corresponds to a possible action the decision subject may take. The root corresponds to taking no action (denoted by $a_{\emptyset}$). Nodes further away from the root on branch $i$ correspond to higher $\Delta \vx_i$, i.e., $\Delta \vx_i(a_{\emptyset}) \prec \Delta \vx_i(a_{i,1}) \prec \ldots \prec \vx_i(a_{i,m_i})$.}
  \label{fig:actionscheme}
\end{figure}
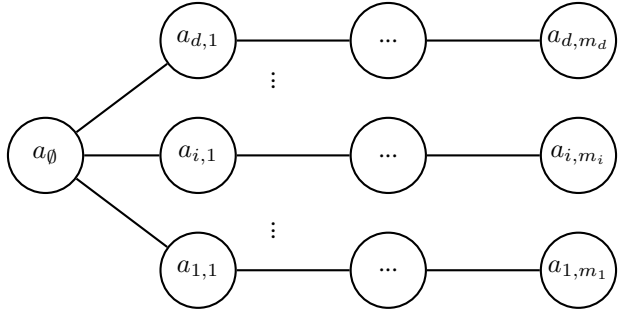

Under this setting, \ref{opt:exact} optimizes over $\Theta(m|\mathcal{R}|)$ variables, where $m$ is the number of actions available to each agent and $|\mathcal{R}|$ is the number of equivalence regions. In order to determine the size of $\mathcal{R}$, we note that an equivalence region can be alternatively characterized by observing that assessment rules $\vtheta$ and $\vtheta'$ belong to the same equivalence region if the difference in their predictions for any two actions $a$ and $a'$ is the same. (This follows from straightforward algebraic manipulation of Definition \ref{def:region}.) As such, an equivalence region $R$ can essentially be characterized by the set of actions $A_R \subseteq \mathcal{A}$ which receive a positive classification when $\vtheta \in R$.\footnote{Specifically, if taking action $a$ results in a positive classification for some $\vtheta \in \vTheta$ and a negative classification for $\vtheta \in \vTheta$, the only way for $\vtheta$ and $\vtheta'$ to be in the same equivalence region is if taking \emph{any} action in $\mathcal{A}$ results in a positive classification for $\vtheta$ and a negative classification for $\vtheta'$. Besides this special case, if $\vtheta$ and $\vtheta'$ result in different classifications for the same action, they are in different equivalence regions.}

Armed with this new characterization of an equivalence region, we are now ready to show the scale of $|\mathcal{R}|$ for the setting described in Figure \ref{fig:actionscheme}. 

\begin{prop}\label{prop:exp}
For the setting described in Figure \ref{fig:actionscheme}, there are $|\mathcal{R}| = \Pi_{i=1}^d m_i - 1$ equivalence regions, where $d$ is the number of observable features of the decision subject and $m_i$ ($\forall i \in [d]$) is the number of actions the decision subject has at their disposal to improve observable feature $i$.
\end{prop}

\begin{proof}
In order to characterize the number of equivalence regions $|\mathcal{R}|$, we define the notion of a \emph{dominated} action $a$, where an action $a$ is dominated by some other action $a'$ if $\Delta \vx(a) \preceq \Delta \vx (a')$, with strict inequality holding for at least one index. Using this notion of dominated actions and our refined characterization of an equivalence region, it is straightforward to see that if action $a$ is dominated by action $a'$, then  $a' \in A_R$ for any equivalence region $R$ where $a \in A_R$. Proposition \ref{prop:exp} then follows directly from the fact that each action only affects one observable feature.
\end{proof}

Proposition \ref{prop:exp} {shows that} the computation of \ref{opt:exact} quickly becomes intractable as the number of observable features grows large, even in this relatively simple setting. This motivates the need for an approximation algorithm for \ref{opt:exact}, which we present in Section \ref{sec:approx_appendix}.

%% file: approx_appendix_full.tex
\section{Proof of Theorem \ref{thm:approx}}\label{sec:approx_appendix}

\begin{algorithm}[t]
        \SetAlgoNoLine
        \KwIn{$\vtheta \in \vTheta$, $\epsilon > 0, \delta > 0$}
        \KwOut{Signaling policy $\widehat{\mathcal{S}} := \{p(\sigma = a|R_{\vtheta})\}_{\forall a \in \mathcal{A}}$ (where  region $R_{\vtheta}$ contains $\vtheta$)} 
        Set $K = \left \lceil \frac{2}{\epsilon^2} \log \left( \frac{2(m^2 + 1)}{\delta}\right) \right \rceil$ \\
        Pick $\ell \in \{1, \ldots, K\}$ uniformly at random. Set $\vtheta_{\ell} = \vtheta$.\\
        Sample $\widetilde{\vTheta} = \{\vtheta_1, \ldots, \vtheta_{\ell-1}, \vtheta_{\ell+1}, \ldots, \vtheta_K\} \sim \pi(\vtheta)$.\\
        Let $\widetilde{\mathcal{R}}$ denote the set of observed regions. Compute $\tilde{p}(R)$, $\forall R \in \widetilde{\mathcal{R}}$, where $\tilde{p}(R)$ is the empirical probability of $\vtheta' \in R$.\\
        Solve
        \small
        \begin{equation}\label{opt:approx}
        \tag*{(APPROX-LP)}
        \begin{aligned}
            \max_{p(\sigma = a|R), \forall a \in \mathcal{A}, R \in \widetilde{\mathcal{R}}} \quad &\sum_{a \in \mathcal{A}} \sum_{R \in \widetilde{\mathcal{R}}} \tilde{p}(R) p(\sigma = a|R) u_{dm}(a)\\
            \text{s.t.} \quad &\sum_{R \in \widetilde{\mathcal{R}}} p(\sigma = a|R)\tilde{p}(R)(u_{ds}(a, R) - u_{ds}(a', R) + \epsilon) \geq 0, \; \forall a, a' \in \mathcal{A}\\
            &\sum_{a \in \mathcal{A}} p(\sigma = a|R) = 1, \; \forall R \in \widetilde{\mathcal{R}}, \; \; \;
            p(\sigma = a|R) \geq 0, \; \forall R \in \widetilde{\mathcal{R}}, a \in \mathcal{A}.\\
        \end{aligned}
        \end{equation}
        \normalsize
        Return signaling policy $\widehat{\mathcal{S}} := \{p(\sigma = a|R_{\vtheta})\}_{\forall a \in \mathcal{A}}$.\\
        \caption{Approximation Algorithm for \ref{opt:exact}}
        \label{alg:approx}
\end{algorithm}

\begin{proof}
First, since the approximation algorithm solves an approximation LP (APPROX-LP) of polynomial size, it runs in polynomial time.

\begin{lemma}\label{lem:poly}
Algorithm \ref{alg:approx} runs in poly($m, \frac{1}{\epsilon}$) time.
\end{lemma}

By bounding the approximation error in the BIC constraints of \ref{opt:approx}, we show that the resulting policy satisfies approximate BIC.

\begin{lemma}\label{lem:eps-BIC}
Algorithm \ref{alg:approx} implements an $\epsilon$-BIC signaling policy.
\end{lemma}

Next, we show that a feasible solution to \ref{opt:approx} exists which achieves expected decision maker utility at least OPT - $\epsilon$ with probability at least $1 - \delta$. In order to do so, we first show that there exists an approximately optimal solution $\widetilde{\mathcal{S}}$ to \ref{opt:exact} such that each signal is either (i) \emph{large} (i.e., output with probability above a certain threshold), or (ii) \emph{honest} (i.e., the signal recommends the action the decision subject would take, had they known the true assessment rule $\vtheta$). Next, we show that $\widetilde{\mathcal{S}}$ is a feasible solution to \ref{opt:approx} with high probability by applying McDiarmid's inequality \cite{mcdiarmid1989method} and a union bound.

\begin{lemma}\label{lem:lh}
There exists an $\frac{\epsilon}{2}$-optimal signaling policy $\widetilde{\mathcal{S}}$ that is large or honest.
\end{lemma}

\begin{lemma}\label{lem:high-prob}
With probability at least $1 - \delta$, $\widetilde{\mathcal{S}}$ is a feasible solution to \ref{opt:approx} and the expected decision maker utility from playing $\widetilde{\mathcal{S}}$ is at least OPT - $\epsilon$.
\end{lemma}

By Lemmas \ref{lem:lh} and \ref{lem:high-prob}, the decision maker's expected utility will be at least OPT - $\epsilon$ with probability at least $1 - \delta$.
\end{proof}

\xhdr{Proof of Lemma \ref{lem:poly}}
\begin{proof} 
Lines 1-3 trivially run in poly($m, \frac{1}{\epsilon}$) time. $\tilde{p}(R)$, $\forall R \in \widetilde{\mathcal{R}}$ can be computed in poly($m, \frac{1}{\epsilon}$) time in an online manner as follows: for each $k \in \{1, \ldots, K\}$, check if $\vtheta_k$ belongs to an existing region. (Note that this can be done in $\mathcal{O}(m)$ time for each region.) If $\vtheta_k$ belongs to an existing region, update the existing empirical probabilities. Otherwise, create a new region. Finally, note that LP \ref{opt:approx} has poly($m, \frac{1}{\epsilon}$) variables and constraints, and can therefore be solved in poly($m, \frac{1}{\epsilon}$) time using, e.g., the Ellipsoid Algorithm.
\end{proof}

\xhdr{Proof of Lemma \ref{lem:eps-BIC}}
\begin{proof}
By the principle of deferred decisions, $\vtheta \sim \Pi'$, where $\Pi'$ is the uniform distribution over $\widetilde{\vTheta}$. \ref{opt:approx} implements an $\epsilon$-BIC signaling policy for $\Pi'$ by definition, so 
\begin{equation*}
    \mathbb{E}_{\vtheta \sim \Pi'}[u_{ds}(a, R) - u_{ds}(a', R) | \sigma = a] \geq -\epsilon, \forall a, a' \in \mathcal{A}.
\end{equation*}
Finally, apply the law of iterated expectation with respect to $\widetilde{\vTheta}$ to obtain the desired result.
\end{proof}

\xhdr{Proof of Lemma \ref{lem:lh}}
In order to prove Lemma \ref{lem:lh}, we make use of the following definitions.

\begin{definition}[Large signal]
A signal $\sigma = a$ is large if $p(\sigma = a) = \sum_{R \in \mathcal{R}} p(\sigma = a|R) p(R) > \frac{\epsilon}{2m}$.
\end{definition}

\begin{definition}[Honest signal]
A signal $\sigma = a$ is honest if $a \in \arg \max_{a' \in \mathcal{A}} u_{ds}(a',R)$.
\end{definition}

\begin{proof} 
We proceed via proof by construction. Let $\mathcal{S}^*$ be the optimal BIC signaling policy. Define the signaling policy $\widetilde{\mathcal{S}}$ as follows: for a given $\vtheta$, it first samples a signal $a \sim \mathcal{S}^*(\vtheta)$. If the signal is \emph{large}, output signal $\sigma = a$. Otherwise, output signal $\sigma = a \in \arg \max_{a' \in \mathcal{A}} u_{ds}(a',R_{\vtheta})$. Every signal of $\widetilde{\mathcal{S}}$ is trivially large or honest. $\widetilde{\mathcal{S}}$ is BIC since $\mathcal{S}^*$ is BIC and $\widetilde{\mathcal{S}}$ only replaces recommendations of $\mathcal{S}^*$ with honest recommendations. Finally, since the total probability of signals that are not large is at most $\frac{\epsilon}{2}$, and the decision maker's utilities are in $[0, 1]$, their expected utility is no worse than $\frac{\epsilon}{2}$ smaller than their expected utility from $\mathcal{S}^*$.
\end{proof}

\xhdr{Proof of Lemma \ref{lem:high-prob}}

The following claim will be useful when proving Lemma \ref{lem:high-prob}.
\begin{claim}\label{clm:exp}
The expected decision maker utility from playing $\widehat{\mathcal{S}}$ is $\sum_{a \in \mathcal{A}} \sum_{R \in \mathcal{R}} p(R) p(\sigma = a|R) u_{dm}(a)$.
\end{claim}

\begin{proof}
The expected decision maker utility from playing $\widehat{\mathcal{S}}$ is $\mathbb{E}_{\vtheta \sim \Pi} [\sum_{a \in \mathcal{A}} \sum_{R \in \widetilde{\mathcal{R}}} \tilde{p}(R) p(\sigma = a|R) u_{dm}(a)] = \mathbb{E}_{\vtheta \sim \Pi} [\sum_{a \in \mathcal{A}} \sum_{R \in \mathcal{R}} \tilde{p}(R) p(\sigma = a|R) u_{dm}(a)]$, by the principle of deferred decisions. Apply the law of iterated expectation with respect to $\widetilde{\vTheta}$ to obtain the desired result.
\end{proof}

Additionally, we will make use of McDiarmid's inequality \cite{mcdiarmid1989method}, stated for completeness below.
\begin{lemma}[McDiarmid's Inequality \cite{mcdiarmid1989method}]\label{lem:mcd}
Let $X_1, \ldots, X_n$ be independent random variables, with $X_k$ taking values in a set $A_k$ for each $k$. Suppose that the (measurable) function $f: \Pi A_k \rightarrow \mathbb{R}$ satisfies 
\begin{equation*}
    |f(\vx) - f(\vx')| \leq c_k
\end{equation*}
whenever the vectors $\vx$ and $\vx'$ differ only in the kth coordinate. Let $Y$ be the random variable $f(X_1, \ldots, X_n)$. Then for any $t > 0$.
\begin{equation*}
    \mathbb{P}(|Y - \mathbb{E}[Y]| \geq t) \leq 2 \exp \left( -2t^2 / \sum_k c_k^2 \right). 
\end{equation*}
\end{lemma}

We are now ready to prove Lemma \ref{lem:high-prob}.

\begin{proof}
First, note that the $\epsilon$-BIC constraints can be rewritten using the observed decision rules as
\begin{equation*}
    \frac{1}{K} \sum_{k=1}^K p(\sigma = a | R_k) (u_{ds}(a,R_k) - u_{ds}(a', R_k)) \geq -\epsilon, \; \; \forall a, a' \in \mathcal{A},
\end{equation*}
where $\vtheta_k \in R_k$. Note that this is a bounded function of $\vtheta_1, \ldots, \vtheta_K$. Let $Y(a, a') = \frac{1}{K} \sum_{k=1}^K p(\sigma = a | R_k) (u_{ds}(a,R_k) - u_{ds}(a', R_k))$. Note that $\mathbb{E}[Y(a, a')] = \sum_{R \in \mathcal{R}} p(\sigma = a|R)p(R)(u_{ds}(a, R) - u_{ds}(a', R))$. 

Applying Lemma \ref{lem:mcd}, we see that $\forall a, a' \in \mathcal{A}$,

\begin{equation*}
    \mathbb{P}(|Y(a,a') - \mathbb{E}[Y(a,a')]| \geq \epsilon) \leq 2\exp \left( -K \epsilon^2 / 2 \right).
\end{equation*}

Similarly, let $Z = \sum_{a \in \mathcal{A}} \sum_{R \in \widetilde{\mathcal{R}}} \tilde{p}(R) p(\sigma = a|R) u_{dm}(a) = \frac{1}{K} \sum_{k=1}^K \sum_{a \in \mathcal{A}} p(\sigma = a | R_k) u_{dm}(a)$ (where $R_k$ contains $\vtheta'_k$). By Claim \ref{clm:exp}, $\mathbb{E}[Z] = \sum_{a \in \mathcal{A}} \sum_{R \in \mathcal{R}} p(R) p(\sigma = a|R) u_{dm}(a)$. Applying Lemma \ref{lem:mcd},

\begin{equation*}
    \mathbb{P}(|Z - \mathbb{E}[Z]| \geq \epsilon/2) \leq 2\exp \left( -K \epsilon^2 /2 \right).
\end{equation*}

Applying the union bound, we see that the probability that \emph{all} $m^2+1$ above inequalities hold is at least $2(m^2 + 1)\exp \left( -K \epsilon^2 /2 \right)$. By inverting the tail bound and picking $K = \frac{2}{\epsilon^2} \log \left( \frac{2(m^2+1)}{\delta}\right)$, we get that $|Z - \mathbb{E}[Z]| \leq \epsilon/2$ \emph{and} $|Y(a,a') - \mathbb{E}[Y(a,a')]| \leq \epsilon$, $\forall a, a' \in \mathcal{A}$, with probability at least $1 - \delta$. Therefore, with probability at least $1 - \delta$, $\widetilde{\mathcal{S}}$ is a feasible solution for LP \ref{opt:approx} and the objective value is at most OPT $- \frac{\epsilon}{2} - \frac{\epsilon}{2} =$ OPT $- \epsilon$. 
\end{proof}

%% file: exp_appendix.tex
\section{Instantiating 1-Dimensional Scenario}
\label{sec:1d-experiment}

In this section we instantiate the example introduced in Section~\ref{sec:ex} and demonstrate the decision maker's gain in utility from the optimal signaling policy over other baselines.
To contextualize this simple synthetic setup, consider a banking institution deciding whether approve a loan application from an applicant based on credit score $\vx_0 \in [300, 850]$ with a simple threshold classifier. The bank approves the application ($\hat{y} = 1$) if $\vx_0 + \theta > 0$ and rejects ($\hat{y} = -1$) otherwise.
Here, we assume the ground-truth threshold value used by the decision maker to be 670 (i.e. $\theta = -670$), which is typically considered as a decent credit score. 
Recall that $a_\emptyset = $ ``do nothing'' and $a_1 = $ ``pay off existing debt'' and set the utility of the decision maker to be $u_{dm}(a_1) = 1, u_{dm}(a_\emptyset) = 0$, as, for the sake of our illustration, we assume credit score to be a good measure of credit-worthiness.
Finally, we assume the prior to be $\pi(\theta) \sim \mathcal{N}(\mu_{\theta}, \sigma_{\theta}^2)$.

\begin{figure}[t]
    \centering
     \begin{subfigure}[b]{0.3\textwidth}
         \centering
         \includegraphics[width=\textwidth]{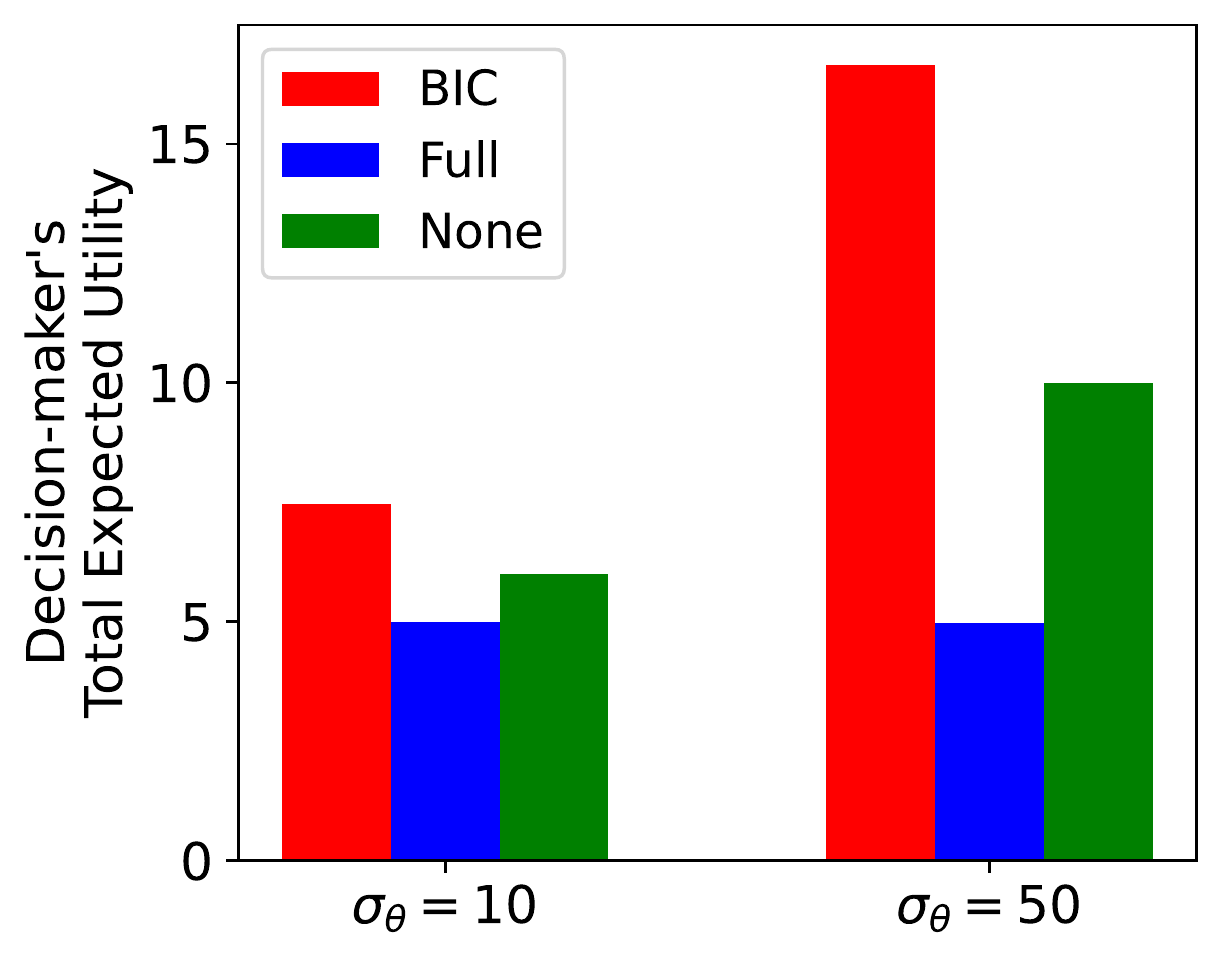}
         \caption{Expected decision maker utility summed up across different $\vx_0$ values under different standard deviation $\sigma_{\theta}$ of the prior.}
         \label{fig:total-utils}
     \end{subfigure}
     \hspace{15pt}
     \begin{subfigure}[b]{0.3\textwidth}
         \centering
         \includegraphics[width=\textwidth]{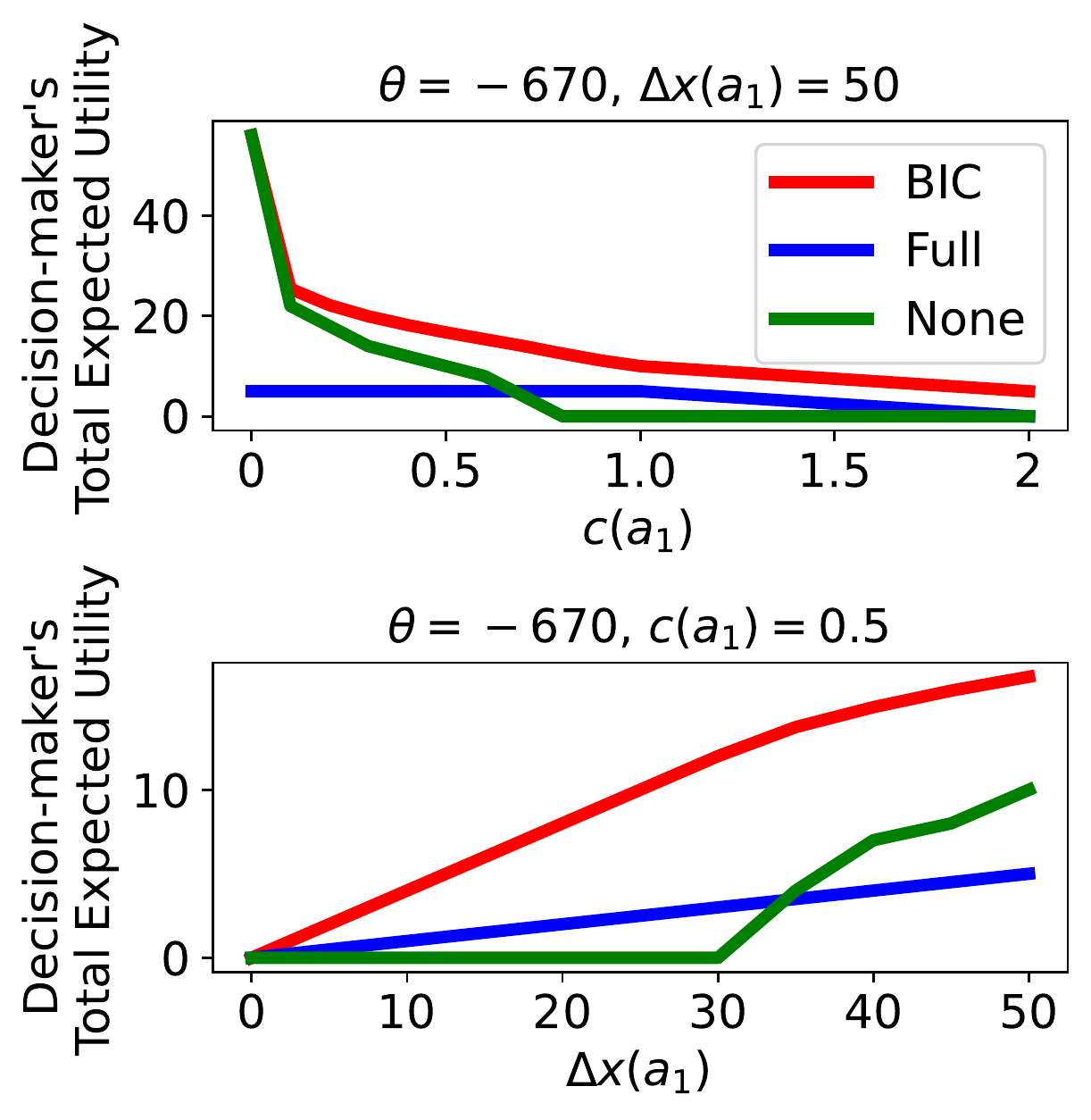}
         \caption{Total expected decision maker utility under different $c(a_1)$ (top) and $\Delta \vx$ (bottom)}
         \label{fig:simple_utility_results_cost_delta}
     \end{subfigure}
    \caption{(a) Total expected decision maker utility summed across difference $\vx_0$ for our optimal signaling policy (\textbf{BIC}, red), against the two baselines: revealing full information about the assessment rule (\textbf{Full}, blue), and revealing no information (\textbf{None}, green). As the decision subject's uncertainty about the true threshold $\theta$ (measured by $\sigma_{\theta}$) increases, the advantage of the optimal signaling policy becomes more visible. (c) When taking action $a_1$ becomes cost-prohibitive (high $c(a_1)$) or less effective (small $\Delta \vx(a_1)$), the decision maker’s utility decreases as there is less incentive for the decision subject to take the action. Nevertheless, the optimal signaling policy yields consistently higher decision maker utility compared to the baselines.}
    \label{fig:simple_utility_results_prior}
\end{figure}

In Figure~\ref{fig:simple_utility_results_prior},
we verify that our optimal signaling policy (\textbf{BIC}, red) yields higher decision maker utility compared to the two baselines: revealing full information (\textbf{Full}, blue) and revealing no information (\textbf{None}, green)\footnote{We set the decision subject cost of taking action $a_1$ to $c(a_1) = 0.5$, and $\Delta \vx(a_1) = 40$ (i.e., action $a_1$ improves an applicant's credit score by $40$ points).}. 
To measure the total amount of decision maker's expected utility yielded by each policy, we assume a uniform distribution of the decision subjects' credit scores $\vx_0$ in the population and take the sum of expected decision maker utility values across different scores. 
We plot these total utility values in Figure~\ref{fig:total-utils}, and as expected, the larger the $\sigma_\theta$ is, the more comparative advantage our method has over the baselines. As the decision subjects' uncertainty about the true $\theta$ increases (i.e., the standard deviation of the prior distribution increases from 10 to 50), the decision maker benefits from our optimal signaling policy even more.



When action $a_1$ becomes more cost-prohibitive (or less effective), as there is less incentive for the decision subjects to take the action, we expect the decision maker's utility to decrease\footnote{In this setting, we set $\mu_{\theta} = -650$ and $\sigma_{\theta}=50$ so that the decision subjects are considered to have a reasonable estimate of the true threshold $\theta=-670$, to make the situation more favorable to the baselines.}.
As shown in Figure~\ref{fig:simple_utility_results_cost_delta}, we indeed observe such a trend as $c(a_1)$ increases (top) and $\Delta \vx(a_1)$ decreases (bottom). 
Nevertheless, our optimal signaling policy yields consistently higher total decision maker utility compared to the baselines across all conditions.

In Figure~\ref{fig:2d-slices}, we show 2-D slices of Figure~\ref{fig:3d-plot} along the $c(a)$ axis (left) and $\Delta \vx (a)$ axis (right).
As is expected, with small cost and sufficiently large $\Delta \vx (a)$ (top row, right), the two baselines become as effective as the optimal signaling policy. Interestingly, we note that changes in different $( c(a_i), \Delta \vx(a_i) )$ result in significantly different rates of change in decision maker utility. For example, the optimal signaling policy (red) and revealing full information (blue) are more resistant to the increase in $c(a_4)$ in range $[0, 0.25]$ than they are for the increase in other $c(a_i)$, $i\neq4$, showing a concave drop in utility rather than a convex one (bottom row, left). Such behavior can be attributed to the relative weight of each feature on the learned assessment rule, where $| \theta_4 | > | \theta_3 | > | \theta_1| > |\theta_2|$. Because the fourth feature has the largest weight, taking action $a_4$ will have the largest effect on an individual's prediction. As a result, the decision maker utility is the least sensitive to increases in the cost of taking that action. Similarly, we observe that the degree to which changes in $\Delta \vx(a)$ affect the expected utility is more drastic for $a_4$ compared to other actions (middle row, right).

\section{Experiment Details and Additional Results}\label{sec:appdx-decision-maker-model}
\begin{table}[t]
    \centering
    \begin{tabular}{c c  c}
        Pair & Feature ($x_i$) & Action ($a_i$) \\ \hline \hline
       $(x_1, a_1)$ &  \# payments with high-utilization ratio & decrease this value \\ \hline
       $(x_2, a_2)$ & \# satisfactory payments & increase this value \\ \hline
       $(x_3, a_3)$ & \% payments that were not delinquent & increase this value \\ \hline
       $(x_4, a_4)$ & revolving balance to credit limit ratio & decrease this value \\ \hline
    \end{tabular}
    \caption{Decision subject's observable features from the HELOC dataset and corresponding actions to improve each feature. For simplicity, we assume that each action only affects one observable feature, although our model generally allows for more intricate relationships between actions and changes in observable features.}
    \label{tab:heloc-details}
\end{table}

\noindent \xhdr{Common prior} We assume the common prior over the realized assessment rule $\vtheta$ takes the form of a multivariate Gaussian $\mathcal{N}(\vtheta, \sigma^2 I_{4 \times 4})$ before training. This captures the setting in which both the decision maker and decision subjects have a good estimate of what the true model will be, but are somewhat uncertain about their estimate. We note that our methods extend to more complicated priors beyond the isotropic Gaussian prior we consider in this setting.

\noindent \xhdr{Changes in observable features} In order to examine the effects that different $\Delta \vx(a_i) (i\in\{1,2,3,4\}$) have on the decision maker's expected utility, we consider settings in which each $\Delta \vx(a_i)$ takes a value in $\{0, 0.25, 0.5, 0.75, 1\}$.

\noindent \xhdr{Utilities and costs of actions} As the decision maker views actions $\{a_1, a_2, a_3, a_4\}$ as equally desirable, we define $u_{dm}(a_i) = 1$, $i \in \{1, 2, 3, 4\}$ and $u_{dm}(a_\emptyset) = 0$.\footnote{We set $u_{dm}(a_1) = u_{dm}(a_2) = u_{dm}(a_3) = u_{dm}(a_4)$ for ease of exposition --- in general, actions can have different utility values based on their relative importance.} 
Since there are 1,320 individuals in our test dataset, the maximum utility the decision maker can obtain is 1,320. As proposed in \cite{rawal2020beyond}, we use the Bradley-Terry model \cite{10.2307/2334029} to generate the decision subject's cost $c(a_i)$ of taking action $a_i$, for $i = 1, 2, 3, 4$. See Appendix~\ref{sec:appdx-cost} for details on our exact generation methods. 

\begin{figure}[t]
    \centering
    \includegraphics[width=0.48\textwidth]{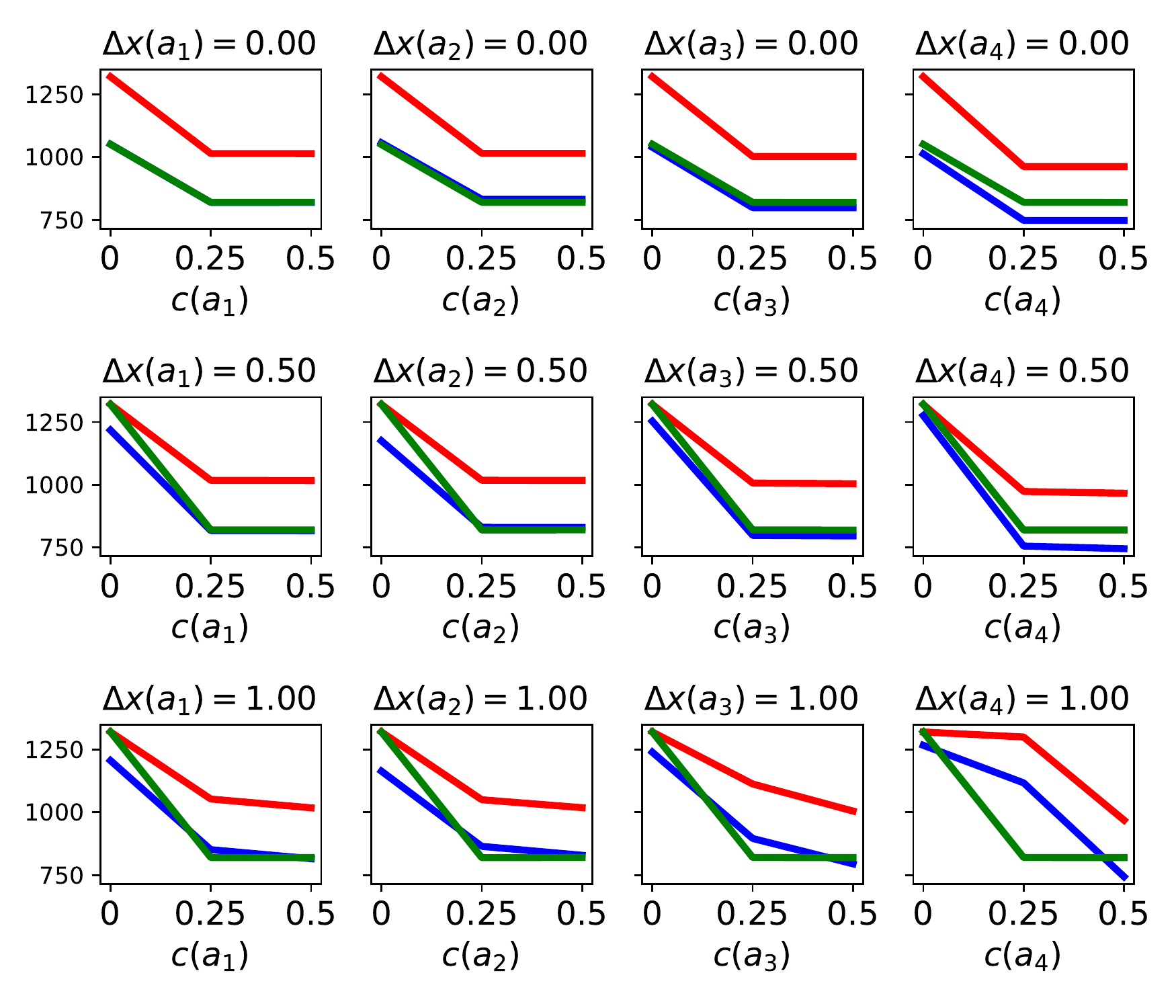}
    \includegraphics[width=0.48\textwidth]{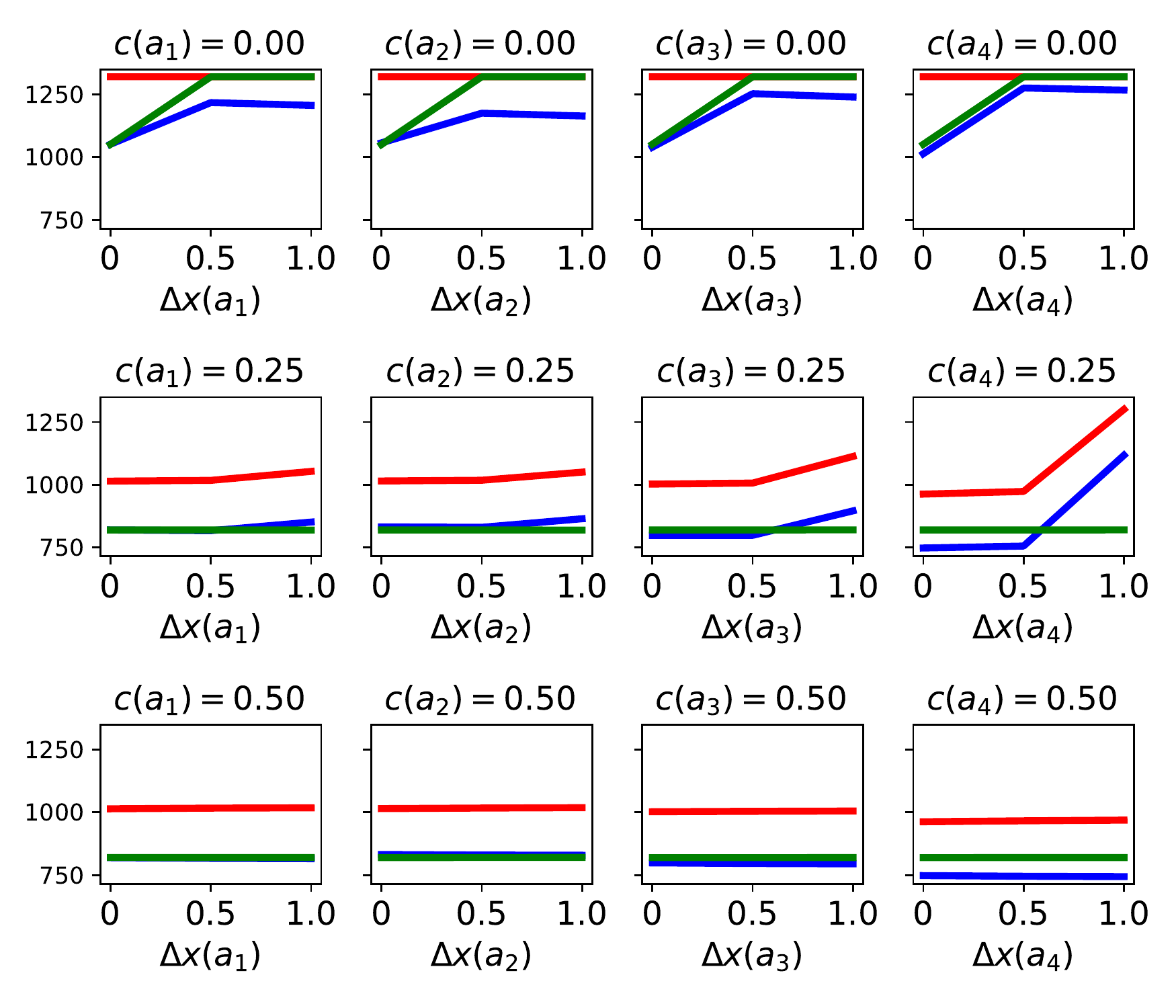}
    \caption{2-D slices of Figure~\ref{fig:3d-plot} across $c(a)$ (left) and $\Delta \vx(a)$ (right). Across these two axes, the optimal signaling policy (red) dominates the revealing full information (blue) and revealing no information (green), though it may be possible for (blue) and (green) to vary in terms of which provides higher decision maker utility.}
    \label{fig:2d-slices}
\end{figure}

\subsection{Remark on the decision maker's assessment rule for HELOC dataset}

To simulate a setting in which the decision maker employs a machine learning model for making decisions about the decision subjects, we train a simple logistic regression model on the subset of HELOC dataset.
We specifically work on four features selected in Table~\ref{tab:heloc-details}, and split the dataset into train/test set (7425, 1857 data points respectively). 
The test accuracy of the model was 71.08 percent, and the corresponding model coefficients were $\vtheta = [-0.22974527,  0.15633134,  0.52023116, -0.61600619]$ with the bias term $-0.08242841$. 
Note that each coefficient term has the sign that is aligned with how the desired action was defined in Table~\ref{tab:heloc-details} (i.e., for features where increasing the value is desirable, the sign is positive and vice-versa). 
To further make sure that the defined actions correctly align with the model, we select the test samples that the trained model made no mistakes on. 
This resulted in a total of 1,320 samples from the test set on which each policy was optimized.

\subsection{Computing different costs for HELOC dataset using Bradley-Terry model}
\label{sec:appdx-cost}

 While exact action costs may be unknown, it is often reasonable for the decision maker to know an \emph{ordering} over possible actions in terms of their cost for decision subjects. For example, it may be common knowledge that opening a new credit card is easier than paying off some existing amount of debt, but exactly \emph{how much} easier may be unclear. The Bradley-Terry model uses exponential score functions to model the probability that feature $x_i$ is more costly for a decision subject to take compared to feature $x_j$. Specifically, it assumes 

\begin{equation*}
    \mathbb{P}(a_i \succ a_j) = \frac{e^{c(a_i)}}{e^{c(a_i)} + e^{c(a_j)}}.
\end{equation*}

Given pairwise cost comparisons (generated from common knowledge or gathered from experts) we can estimate $\mathbb{P}(a_i \succ a_j)$ and solve for the parameters $\{c(a_i)\}_{i=1}^4$ using maximum likelihood estimation. In order to gain more insight into how different action cost orderings affect the decision maker utility, we consider several different ground-truth cost orderings over actions and simulate expert advice in order to estimate $\mathbb{P}(a_i \succ a_j)$, $\forall a_i, a_j \in \mathcal{A}$. While the expert advice is purely synthetic in our setting, this method provides a principled way to estimate action costs whenever input from domain experts (e.g., financial advisors) is available to the decision maker.

We use the following set of comparison inputs (manually generated) in Table~\ref{tab:cost1}-\ref{tab:cost4} to generate cost values with the relative ordering presented in Section~\ref{sec:experiments}. 
While these comparison inputs are generated arbitrarily for the simulations, these can be obtained by querying several domain experts and aggregating their answers regarding which feature is more difficult to change.   
The resulting cost values are shown in Table~\ref{tab:cost-values}.

\begin{table}[]
    \begin{subtable}[t]{0.48\textwidth}
    \begin{tabular}[t]{c c c c}
        Feature A & Feature B & \# (A $>$ B) & \# (A $<$ B) \\ \hline \hline
        $\vx_1$ & $\vx_2$ & 8 & 2 \\ \hline
        $\vx_1$ & $\vx_3$ & 9 & 1 \\ \hline
        $\vx_1$ & $\vx_4$ & 7 & 3 \\ \hline
        $\vx_2$ & $\vx_3$ & 2 & 8 \\ \hline
        $\vx_2$ & $\vx_4$ & 0 & 10 \\ \hline
        $\vx_3$ & $\vx_4$ & 1 & 9 \\ \hline
    \end{tabular}
    \caption{$c(a_1) > c(a_4) > c(a_3) > c(a_2)$}
    \label{tab:cost1}
    \end{subtable}
    \hspace{\fill}
    \begin{subtable}[t]{0.48\textwidth}
    \begin{tabular}[t]{c c c c}
        Feature A & Feature B & \# (A $>$ B) & \# (A $<$ B) \\ \hline \hline
        $\vx_1$ & $\vx_2$ & 2 & 8 \\ \hline
        $\vx_1$ & $\vx_3$ & 3 & 7 \\ \hline
        $\vx_1$ & $\vx_4$ & 4 & 6 \\ \hline
        $\vx_2$ & $\vx_3$ & 6 & 4 \\ \hline
        $\vx_2$ & $\vx_4$ & 7 & 3 \\ \hline
        $\vx_3$ & $\vx_4$ & 6 & 4 \\ \hline
    \end{tabular}
    \caption{$c(a_2) > c(a_3) > c(a_4) > c(a_1)$}
    \label{tab:cost2}
    \end{subtable}
    \begin{subtable}[t]{0.48\textwidth}
    \begin{tabular}[t]{c c c c}
        Feature A & Feature B & \# (A $>$ B) & \# (A $<$ B) \\ \hline \hline
        $\vx_1$ & $\vx_2$ & 2 & 8 \\ \hline
        $\vx_1$ & $\vx_3$ & 1 & 9 \\ \hline
        $\vx_1$ & $\vx_4$ & 4 & 6 \\ \hline
        $\vx_2$ & $\vx_3$ & 3 & 7 \\ \hline
        $\vx_2$ & $\vx_4$ & 7 & 3 \\ \hline
        $\vx_3$ & $\vx_4$ & 7 & 3 \\ \hline
    \end{tabular}
    \caption{$c(a_3) > c(a_2) > c(a_4) > c(a_1)$}
    \label{tab:cost3}
    \end{subtable}
    \hspace{\fill}
    \begin{subtable}[t]{0.48\textwidth}
    \begin{tabular}[t]{c c c c}
        Feature A & Feature B & \# (A $>$ B) & \# (A $<$ B) \\ \hline \hline
        $\vx_1$ & $\vx_2$ & 8 & 2 \\ \hline
        $\vx_1$ & $\vx_3$ & 9 & 1 \\ \hline
        $\vx_1$ & $\vx_4$ & 2 & 3 \\ \hline
        $\vx_2$ & $\vx_3$ & 7 & 8 \\ \hline
        $\vx_2$ & $\vx_4$ & 0 & 10 \\ \hline
        $\vx_3$ & $\vx_4$ & 1 & 9 \\ \hline
    \end{tabular}
    \caption{$c(a_4) > c(a_1) > c(a_3) > c(a_2)$}
    \label{tab:cost4}
    \end{subtable}
    \begin{subtable}[t]{0.7\textwidth}
    \centering
    \begin{tabular}[t]{c c c c c}
        Configuration & $c(a_1)$ & $c(a_2)$ & $c(a_3)$ & $c(a_4)$\\ \hline \hline
        (i)  &  0.5151 & 0.0282 & 0.0723 & 0.3844 \\ \hline
        (ii) &  0.1159 &  0.428  &  0.2758 & 0.1803 \\ \hline
        (iii) &  0.07640764 & 0.27692769 & 0.50635064 & 0.14031403 \\ \hline
        (iv)  &  0.2987 & 0.0428 & 0.0476 & 0.6109\\ \hline
    \end{tabular}
    \caption{Cost values learned by the Bradley-Terry model from the pair-wise comparison inputs above.}
    \label{tab:cost-values}
    \end{subtable}
    \caption{Comparison inputs used by the Bradley-Terry model to generate different cost configurations.}
    \label{tab:cost-comparison-inputs}
\end{table}

\subsection{Additional results for different cost and $\Delta \vx(a)$ configurations}
\label{sec:appdx-additional-results}

Figure \ref{fig:additional-details-configs} shows more exhaustive results on different cost configurations (i)-(iv) as defined in Table~\ref{tab:cost-values} and $\Delta \vx(a_i) \in {0, 0.25, 0.5, 0.75, 1.0}$ for $i=1,2,3,4$ on HELOC datset. For all configurations considered, our optimal signaling policy (red) consistently yields utility no less than both baselines: revealing full information about the assessment rule (blue), and revealing no information (green). 

\begin{figure}
    \centering
    \includegraphics[width=0.48\textwidth]{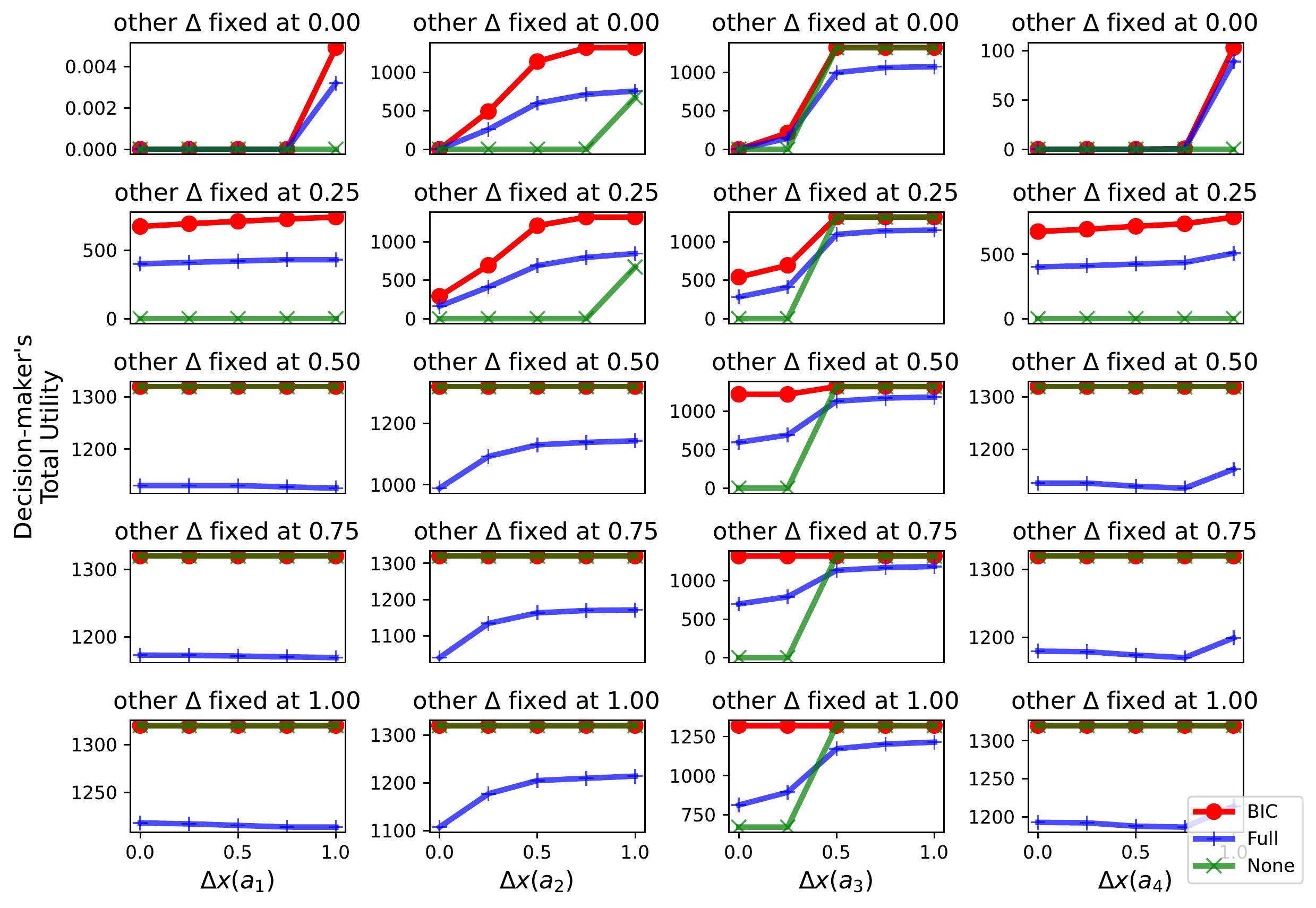}
    \includegraphics[width=0.48\textwidth]{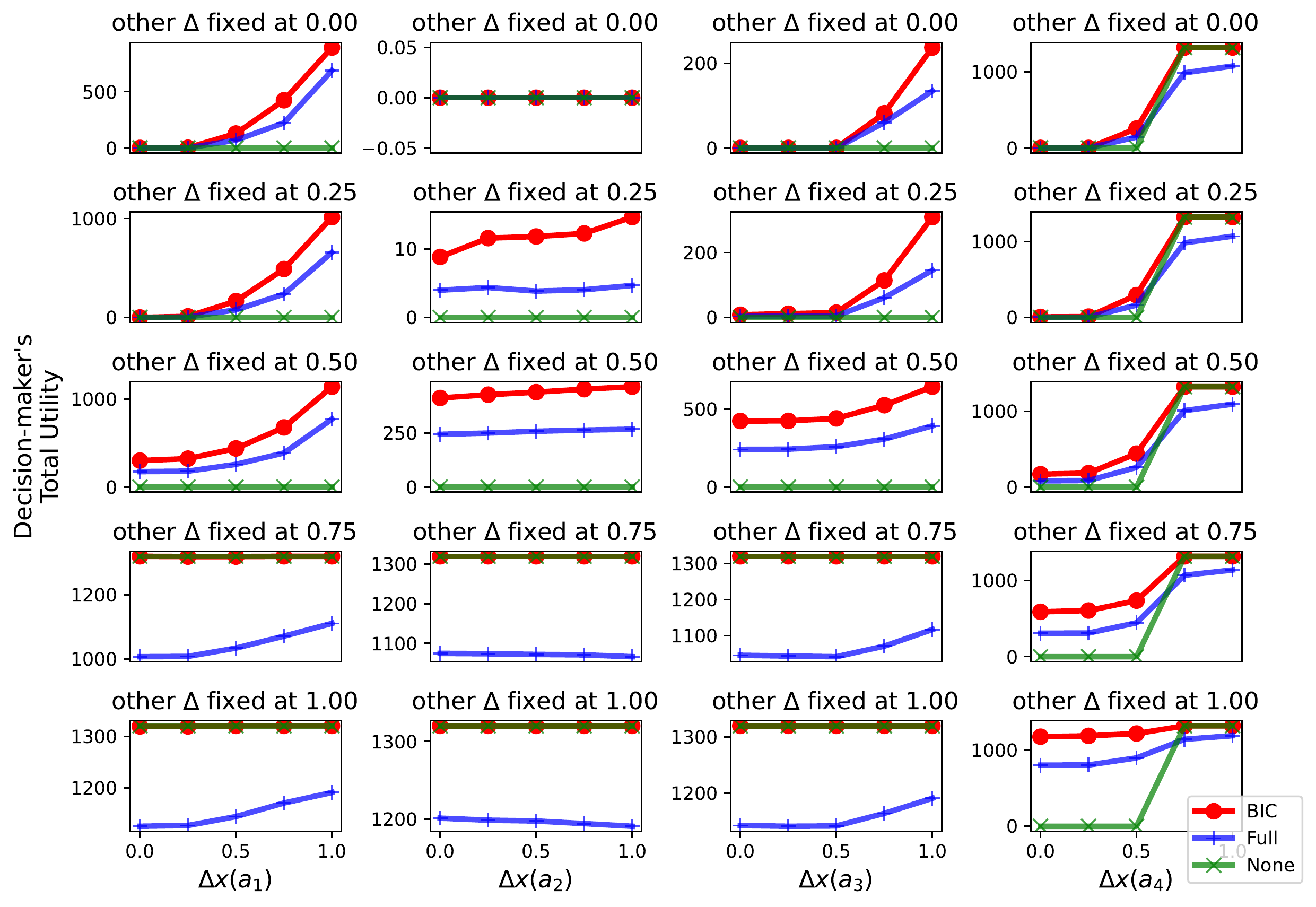}
    \includegraphics[width=0.48\textwidth]{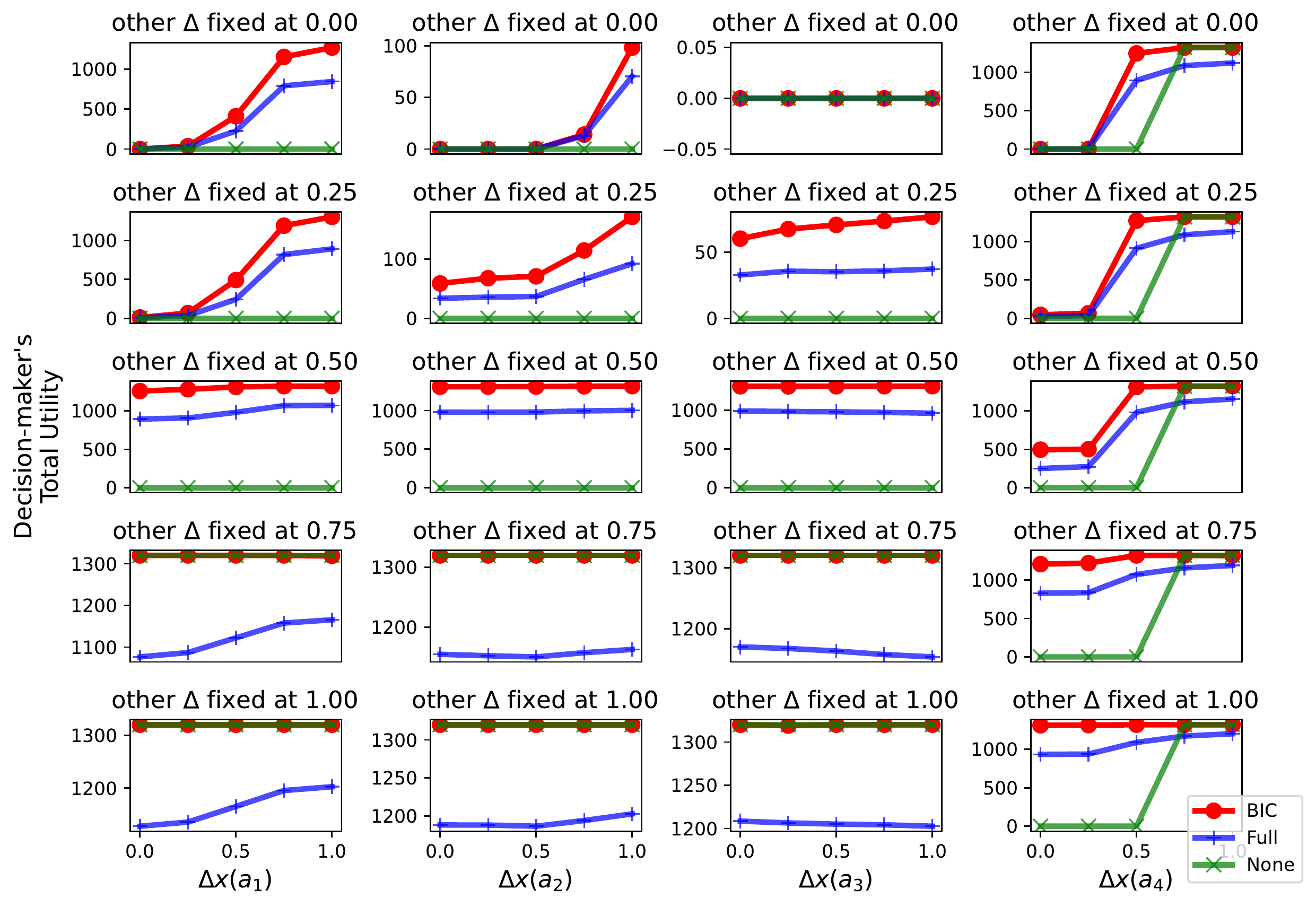}
    \includegraphics[width=0.48\textwidth]{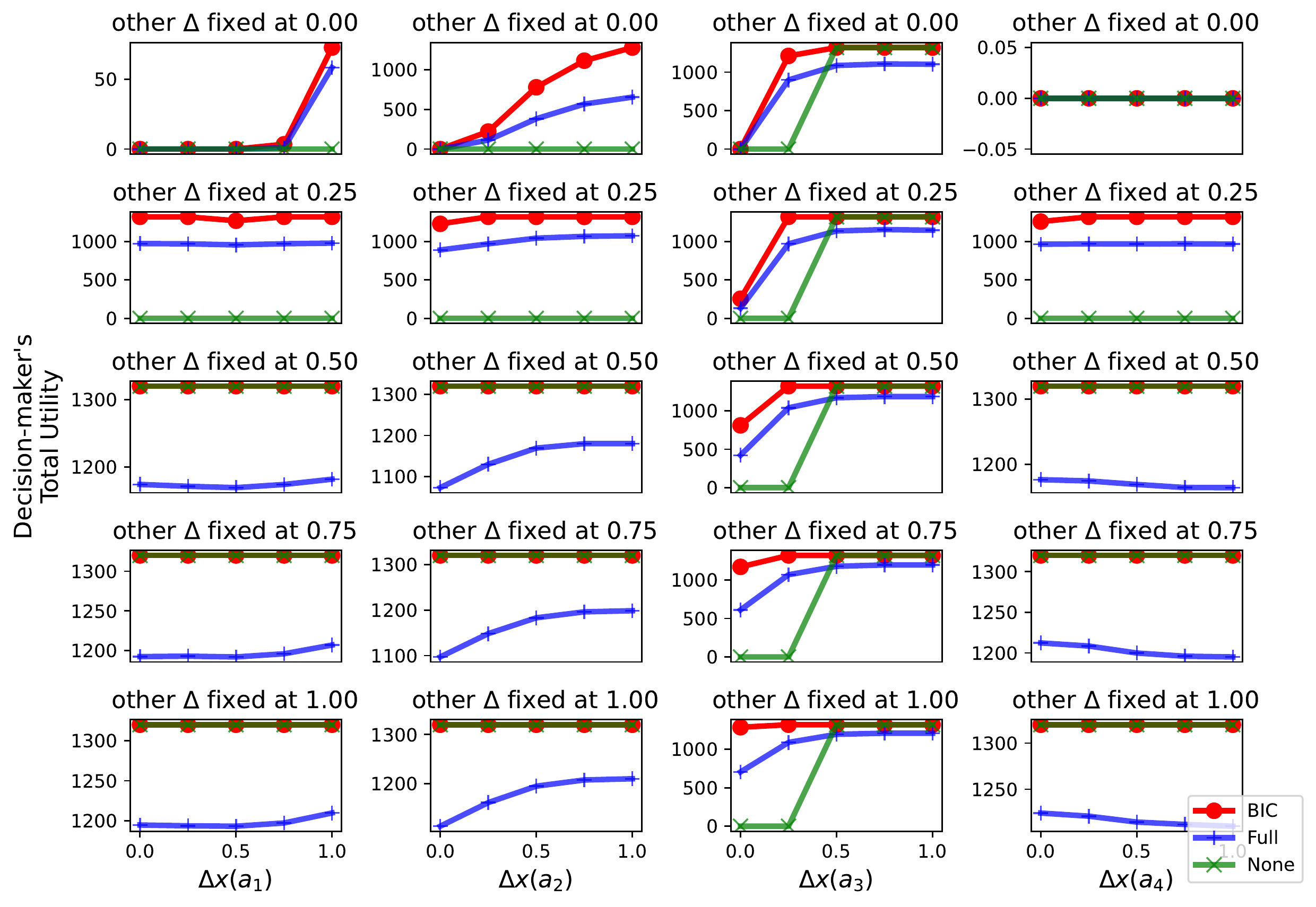}
    \caption{More detailed view on decision maker utility for different $\Delta \vx (a)$ values and cost configurations (i)-(iv) (from upper right to bottom right quadrants). Our optimal signaling policy (red) consistently achieves utility no less than revealing full information (blue) and revealing no information (green) in all settings.}
    \label{fig:additional-details-configs}
\end{figure}

\subsection{Additional results for different utility configurations}

Previously we viewed actions $\{a_1, a_2, a_3, a_4\}$ as equally desirable to the decision maker ($u_{dm}(a_i) = 1$ for $i=1,2,3,4$). 
Figure~\ref{fig:3d-plot-different-utils} and \ref{fig:2d-slices-different-utils} show the results with a different utility profile, $u_{dm}(a_1) = 0.25, u_{dm}(a_2) = 0.5, u_{dm}(a_3) = 0.75, u_{dm}(a_4) = 1$, where each action gives different utility for the decision maker. 
Similar to what we previously observed before, the optimal signaling policy (red) effectively upper-bounds the two baselines, revealing everything (blue) and revealing nothing (green) in all settings. 
As expected, when certain action has higher utility (e.g. $a_4$ compared to $a_1$), the total utility becomes more sensitive to changes in $c(a)$ and $\Delta \mathbf{x}(a)$.

\begin{figure}
    \centering
    \includegraphics[width=\textwidth]{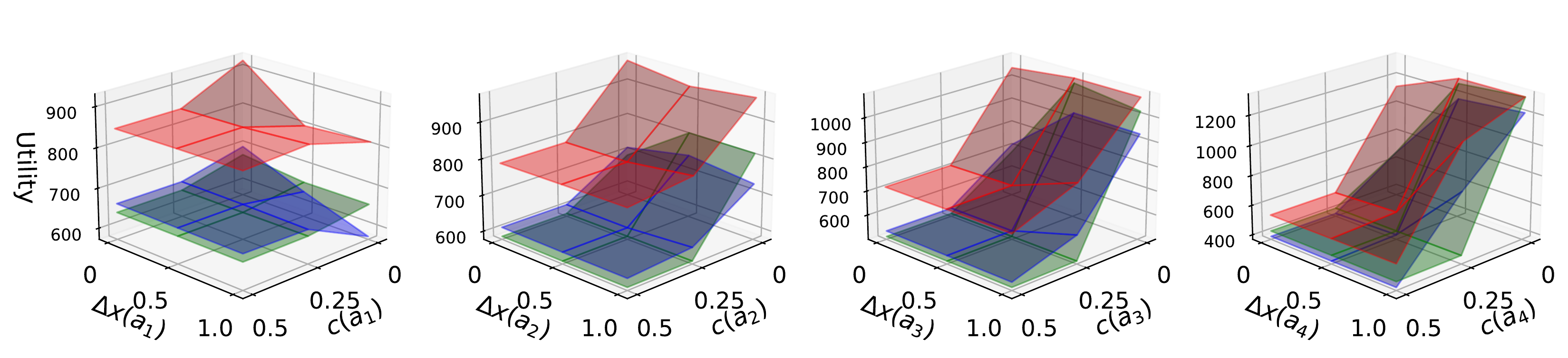}
    \caption{Surface of expected utility across different $c(a)$ and $\Delta \mathbf{x}(a)$ configurations (same setting as in Figure~\ref{fig:3d-plot}, but with different decision maker utility for each action).}
    \label{fig:3d-plot-different-utils}
\end{figure}

\begin{figure}
    \centering
    \includegraphics[width=0.49\textwidth]{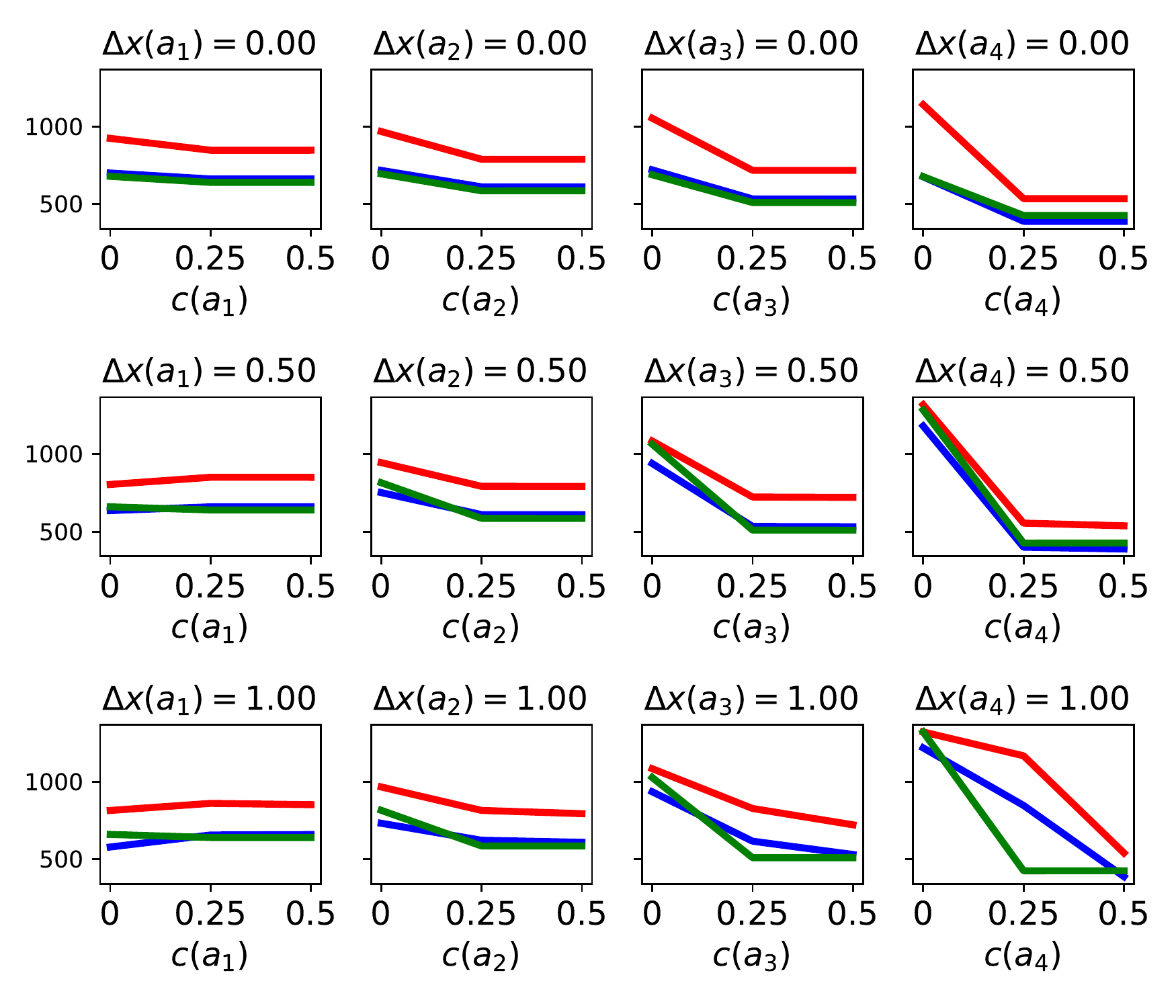}
    \includegraphics[width=0.49\textwidth]{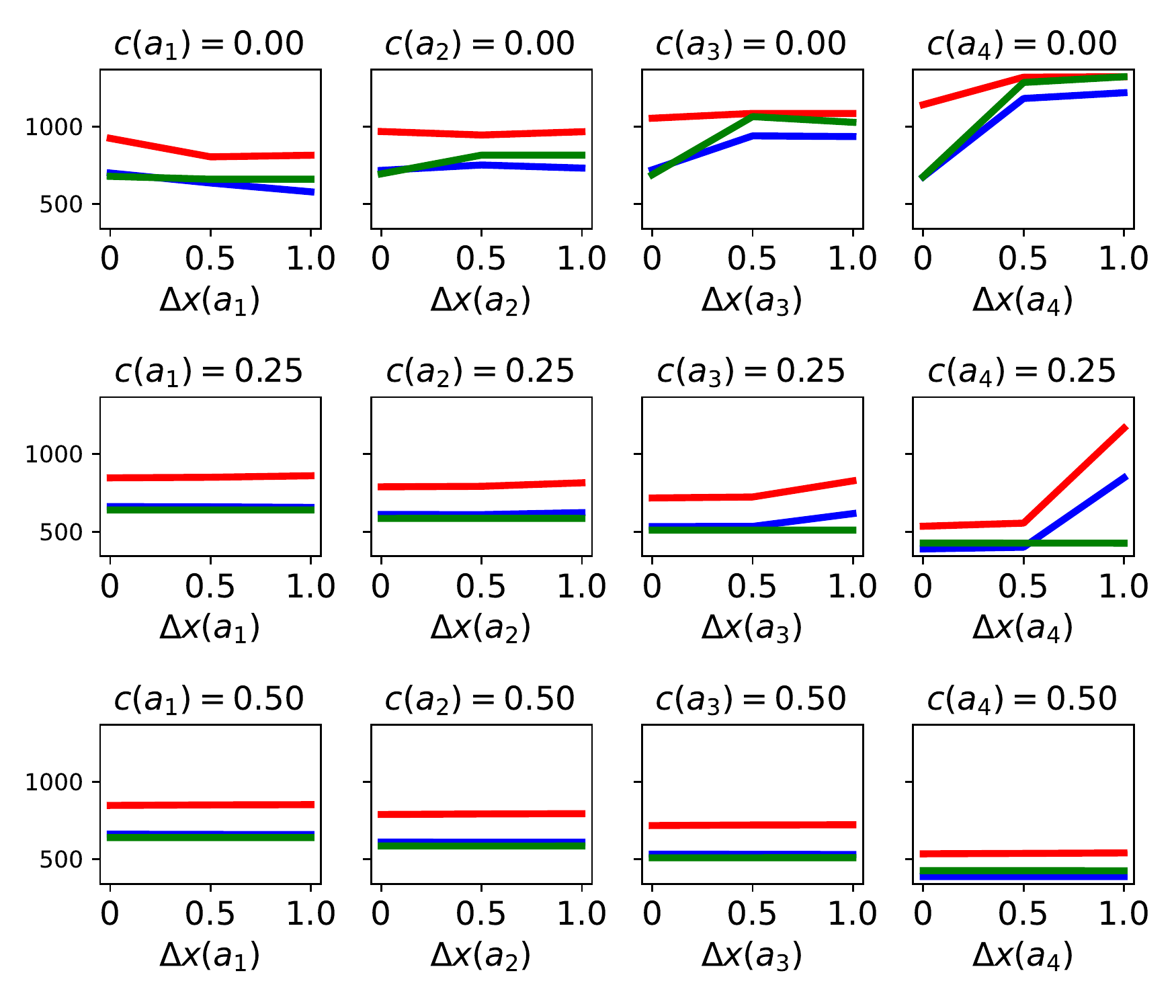}
    \caption{2-D slices of Figure~\ref{fig:3d-plot-different-utils} across $c(a)$ (left) and $\Delta \mathbf{x}(a)$ (right).}
    \label{fig:2d-slices-different-utils}
\end{figure}

%% file: main.bbl
\begin{thebibliography}{43}
\providecommand{\natexlab}[1]{#1}
\providecommand{\url}[1]{\texttt{#1}}
\expandafter\ifx\csname urlstyle\endcsname\relax
  \providecommand{\doi}[1]{doi: #1}\else
  \providecommand{\doi}{doi: \begingroup \urlstyle{rm}\Url}\fi

\bibitem[Akyol et~al.(2016)Akyol, Langbort, and Basar]{akyol2016price}
E.~Akyol, C.~Langbort, and T.~Basar.
\newblock Price of transparency in strategic machine learning.
\newblock \emph{CoRR}, abs/1610.08210, 2016.
\newblock URL \url{http://arxiv.org/abs/1610.08210}.

\bibitem[Bechavod et~al.(2022)Bechavod, Podimata, Wu, and
  Ziani]{bechavod2021information}
Y.~Bechavod, C.~Podimata, Z.~S. Wu, and J.~Ziani.
\newblock Information discrepancy in strategic learning.
\newblock In K.~Chaudhuri, S.~Jegelka, L.~Song, C.~Szepesv{\'{a}}ri, G.~Niu,
  and S.~Sabato, editors, \emph{International Conference on Machine Learning,
  {ICML} 2022, 17-23 July 2022, Baltimore, Maryland, {USA}}, volume 162 of
  \emph{Proceedings of Machine Learning Research}, pages 1691--1715. {PMLR},
  2022.
\newblock URL \url{https://proceedings.mlr.press/v162/bechavod22a.html}.

\bibitem[Bradley and Terry(1952)]{10.2307/2334029}
R.~A. Bradley and M.~E. Terry.
\newblock Rank analysis of incomplete block designs: I. the method of paired
  comparisons.
\newblock \emph{Biometrika}, 39\penalty0 (3/4):\penalty0 324--345, 1952.
\newblock ISSN 00063444.
\newblock URL \url{http://www.jstor.org/stable/2334029}.

\bibitem[Castiglioni et~al.(2020)Castiglioni, Celli, Marchesi, and
  Gatti]{castiglioni2020online}
M.~Castiglioni, A.~Celli, A.~Marchesi, and N.~Gatti.
\newblock Online bayesian persuasion.
\newblock In H.~Larochelle, M.~Ranzato, R.~Hadsell, M.~Balcan, and H.~Lin,
  editors, \emph{Advances in Neural Information Processing Systems 33: Annual
  Conference on Neural Information Processing Systems 2020, NeurIPS 2020,
  December 6-12, 2020, virtual}, 2020.
\newblock URL
  \url{https://proceedings.neurips.cc/paper/2020/hash/ba5451d3c91a0f982f103cdbe249bc78-Abstract.html}.

\bibitem[Chalfin et~al.(2016)Chalfin, Danieli, Hillis, Jelveh, Luca, Ludwig,
  and Mullainathan]{chalfin2016productivity}
A.~Chalfin, O.~Danieli, A.~Hillis, Z.~Jelveh, M.~Luca, J.~Ludwig, and
  S.~Mullainathan.
\newblock Productivity and selection of human capital with machine learning.
\newblock \emph{American Economic Review}, 106\penalty0 (5):\penalty0 124--27,
  May 2016.
\newblock \doi{10.1257/aer.p20161029}.
\newblock URL \url{https://www.aeaweb.org/articles?id=10.1257/aer.p20161029}.

\bibitem[Chen et~al.(2018)Chen, Frazier, and Kempe]{chen2018incentivizing}
B.~Chen, P.~I. Frazier, and D.~Kempe.
\newblock Incentivizing exploration by heterogeneous users.
\newblock In S.~Bubeck, V.~Perchet, and P.~Rigollet, editors, \emph{Conference
  On Learning Theory, {COLT} 2018, Stockholm, Sweden, 6-9 July 2018}, volume~75
  of \emph{Proceedings of Machine Learning Research}, pages 798--818. {PMLR},
  2018.
\newblock URL \url{http://proceedings.mlr.press/v75/chen18a.html}.

\bibitem[Chen et~al.(2022)Chen, Li, Kim, Plumb, and Talwalkar]{chen2021towards}
V.~Chen, J.~Li, J.~S. Kim, G.~Plumb, and A.~Talwalkar.
\newblock Interpretable machine learning: moving from mythos to diagnostics.
\newblock \emph{Commun. {ACM}}, 65\penalty0 (8):\penalty0 43--50, 2022.
\newblock \doi{10.1145/3546036}.
\newblock URL \url{https://doi.org/10.1145/3546036}.

\bibitem[Chen et~al.(2021)Chen, Wang, and Liu]{chen2021strategic}
Y.~Chen, J.~Wang, and Y.~Liu.
\newblock Strategic classification with a light touch: Learning classifiers
  that incentivize constructive adaptation, 2021.

\bibitem[Citron and Pasquale(2014)]{citron2014scored}
D.~Citron and F.~Pasquale.
\newblock The scored society: Due process for automated predictions.
\newblock \emph{Washington Law Review}, 89:\penalty0 1--33, 03 2014.

\bibitem[{Council of European Union}(2016)]{gdpr}
{Council of European Union}.
\newblock Council regulation ({EU}) no 679/2016, 2016.
\newblock
  \newline\url{https://eur-lex.europa.eu/legal-content/EN/TXT/PDF/?uri=CELEX:32016R0679}.

\bibitem[Dong et~al.(2018)Dong, Roth, Schutzman, Waggoner, and Wu]{dongetal}
J.~Dong, A.~Roth, Z.~Schutzman, B.~Waggoner, and Z.~S. Wu.
\newblock Strategic classification from revealed preferences.
\newblock In {\'{E}}.~Tardos, E.~Elkind, and R.~Vohra, editors,
  \emph{Proceedings of the 2018 {ACM} Conference on Economics and Computation,
  Ithaca, NY, USA, June 18-22, 2018}, pages 55--70. {ACM}, 2018.
\newblock \doi{10.1145/3219166.3219193}.
\newblock URL \url{https://doi.org/10.1145/3219166.3219193}.

\bibitem[Dughmi and Xu(2017)]{dughmi2017algorithmic2}
S.~Dughmi and H.~Xu.
\newblock Algorithmic persuasion with no externalities.
\newblock In C.~Daskalakis, M.~Babaioff, and H.~Moulin, editors,
  \emph{Proceedings of the 2017 {ACM} Conference on Economics and Computation,
  {EC} '17, Cambridge, MA, USA, June 26-30, 2017}, pages 351--368. {ACM}, 2017.
\newblock \doi{10.1145/3033274.3085152}.
\newblock URL \url{https://doi.org/10.1145/3033274.3085152}.

\bibitem[Dughmi and Xu(2021)]{dughmi2019algorithmic}
S.~Dughmi and H.~Xu.
\newblock Algorithmic bayesian persuasion.
\newblock \emph{{SIAM} J. Comput.}, 50\penalty0 (3), 2021.
\newblock \doi{10.1137/16M1098334}.
\newblock URL \url{https://doi.org/10.1137/16M1098334}.

\bibitem[Dworczak and Martini(2019)]{dworczak2019simple}
P.~Dworczak and G.~Martini.
\newblock The simple economics of optimal persuasion.
\newblock \emph{Journal of Political Economy}, 127\penalty0 (5):\penalty0
  1993--2048, 2019.
\newblock \doi{10.1086/701813}.
\newblock URL \url{https://doi.org/10.1086/701813}.

\bibitem[FICO(2018)]{FICO}
FICO.
\newblock Explainable machine learning challenge.
\newblock
  \url{https://community.fico.com/s/explainable-machine-learning-challenge},
  2018.

\bibitem[Frankel and Kartik(2021)]{frankel2019improving}
A.~Frankel and N.~Kartik.
\newblock {Improving Information from Manipulable Data}.
\newblock \emph{Journal of the European Economic Association}, 20\penalty0
  (1):\penalty0 79--115, 06 2021.
\newblock ISSN 1542-4766.
\newblock \doi{10.1093/jeea/jvab017}.
\newblock URL \url{https://doi.org/10.1093/jeea/jvab017}.

\bibitem[Ghalme et~al.(2021)Ghalme, Nair, Eilat, Talgam{-}Cohen, and
  Rosenfeld]{ghalme2021strategic}
G.~Ghalme, V.~Nair, I.~Eilat, I.~Talgam{-}Cohen, and N.~Rosenfeld.
\newblock Strategic classification in the dark.
\newblock In M.~Meila and T.~Zhang, editors, \emph{Proceedings of the 38th
  International Conference on Machine Learning, {ICML} 2021, 18-24 July 2021,
  Virtual Event}, volume 139 of \emph{Proceedings of Machine Learning
  Research}, pages 3672--3681. {PMLR}, 2021.
\newblock URL \url{http://proceedings.mlr.press/v139/ghalme21a.html}.

\bibitem[Hardt et~al.(2016)Hardt, Megiddo, Papadimitriou, and
  Wootters]{hardt2016strategic}
M.~Hardt, N.~Megiddo, C.~H. Papadimitriou, and M.~Wootters.
\newblock Strategic classification.
\newblock In M.~Sudan, editor, \emph{Proceedings of the 2016 {ACM} Conference
  on Innovations in Theoretical Computer Science, Cambridge, MA, USA, January
  14-16, 2016}, pages 111--122. {ACM}, 2016.
\newblock \doi{10.1145/2840728.2840730}.
\newblock URL \url{https://doi.org/10.1145/2840728.2840730}.

\bibitem[Harris et~al.(2021)Harris, Heidari, and Wu]{harris2021stateful}
K.~Harris, H.~Heidari, and Z.~S. Wu.
\newblock Stateful strategic regression.
\newblock In M.~Ranzato, A.~Beygelzimer, Y.~N. Dauphin, P.~Liang, and J.~W.
  Vaughan, editors, \emph{Advances in Neural Information Processing Systems 34:
  Annual Conference on Neural Information Processing Systems 2021, NeurIPS
  2021, December 6-14, 2021, virtual}, pages 28728--28741, 2021.
\newblock URL
  \url{https://proceedings.neurips.cc/paper/2021/hash/f1404c2624fa7f2507ba04fd9dfc5fb1-Abstract.html}.

\bibitem[Harris et~al.(2022)Harris, Ngo, Stapleton, Heidari, and
  Wu]{harris2021strategic}
K.~Harris, D.~D.~T. Ngo, L.~Stapleton, H.~Heidari, and S.~Wu.
\newblock Strategic instrumental variable regression: Recovering causal
  relationships from strategic responses.
\newblock In K.~Chaudhuri, S.~Jegelka, L.~Song, C.~Szepesv{\'{a}}ri, G.~Niu,
  and S.~Sabato, editors, \emph{International Conference on Machine Learning,
  {ICML} 2022, 17-23 July 2022, Baltimore, Maryland, {USA}}, volume 162 of
  \emph{Proceedings of Machine Learning Research}, pages 8502--8522. {PMLR},
  2022.
\newblock URL \url{https://proceedings.mlr.press/v162/harris22a.html}.

\bibitem[Homonoff et~al.(2021)Homonoff, O'Brien, and Sussman]{homonoff2021does}
T.~Homonoff, R.~O'Brien, and A.~B. Sussman.
\newblock Does knowing your fico score change financial behavior? evidence from
  a field experiment with student loan borrowers.
\newblock \emph{Review of Economics and Statistics}, 103\penalty0 (2):\penalty0
  236--250, 2021.

\bibitem[Immorlica et~al.(2019)Immorlica, Mao, Slivkins, and
  Wu]{immorlica2019bayesian}
N.~Immorlica, J.~Mao, A.~Slivkins, and Z.~S. Wu.
\newblock Bayesian exploration with heterogeneous agents.
\newblock In L.~Liu, R.~W. White, A.~Mantrach, F.~Silvestri, J.~J. McAuley,
  R.~Baeza{-}Yates, and L.~Zia, editors, \emph{The World Wide Web Conference,
  {WWW} 2019, San Francisco, CA, USA, May 13-17, 2019}, pages 751--761. {ACM},
  2019.
\newblock \doi{10.1145/3308558.3313649}.
\newblock URL \url{https://doi.org/10.1145/3308558.3313649}.

\bibitem[Jagadeesan et~al.(2021)Jagadeesan, Mendler{-}D{\"{u}}nner, and
  Hardt]{jagadeesan2021alternative}
M.~Jagadeesan, C.~Mendler{-}D{\"{u}}nner, and M.~Hardt.
\newblock Alternative microfoundations for strategic classification.
\newblock In M.~Meila and T.~Zhang, editors, \emph{Proceedings of the 38th
  International Conference on Machine Learning, {ICML} 2021, 18-24 July 2021,
  Virtual Event}, volume 139 of \emph{Proceedings of Machine Learning
  Research}, pages 4687--4697. {PMLR}, 2021.
\newblock URL \url{http://proceedings.mlr.press/v139/jagadeesan21a.html}.

\bibitem[Jagtiani and Lemieux(2019)]{jagtiani2019roles}
J.~Jagtiani and C.~Lemieux.
\newblock The roles of alternative data and machine learning in fintech
  lending: Evidence from the lendingclub consumer platform.
\newblock \emph{Financial Management}, 48\penalty0 (4):\penalty0 1009--1029,
  2019.
\newblock \doi{https://doi.org/10.1111/fima.12295}.
\newblock URL \url{https://onlinelibrary.wiley.com/doi/abs/10.1111/fima.12295}.

\bibitem[Joshi et~al.(2019)Joshi, Koyejo, Vijitbenjaronk, Kim, and
  Ghosh]{joshi2019towards}
S.~Joshi, O.~Koyejo, W.~Vijitbenjaronk, B.~Kim, and J.~Ghosh.
\newblock Towards realistic individual recourse and actionable explanations in
  black-box decision making systems.
\newblock \emph{CoRR}, abs/1907.09615, 2019.
\newblock URL \url{http://arxiv.org/abs/1907.09615}.

\bibitem[Kamenica(2019)]{kamenica2019bayesian}
E.~Kamenica.
\newblock Bayesian persuasion and information design.
\newblock \emph{Annual Review of Economics}, 11\penalty0 (1):\penalty0
  249--272, 2019.
\newblock \doi{10.1146/annurev-economics-080218-025739}.
\newblock URL \url{https://doi.org/10.1146/annurev-economics-080218-025739}.

\bibitem[Kamenica and Gentzkow(2011)]{kamenica2011bayesian}
E.~Kamenica and M.~Gentzkow.
\newblock Bayesian persuasion.
\newblock \emph{American Economic Review}, 101\penalty0 (6):\penalty0
  2590--2615, October 2011.
\newblock \doi{10.1257/aer.101.6.2590}.
\newblock URL \url{https://www.aeaweb.org/articles?id=10.1257/aer.101.6.2590}.

\bibitem[Karimi et~al.(2020)Karimi, Barthe, Sch{\"{o}}lkopf, and
  Valera]{karimi2021survey}
A.~Karimi, G.~Barthe, B.~Sch{\"{o}}lkopf, and I.~Valera.
\newblock A survey of algorithmic recourse: definitions, formulations,
  solutions, and prospects.
\newblock \emph{CoRR}, abs/2010.04050, 2020.
\newblock URL \url{https://arxiv.org/abs/2010.04050}.

\bibitem[Kleinberg and Raghavan(2020)]{kleinberg2020classifiers}
J.~M. Kleinberg and M.~Raghavan.
\newblock How do classifiers induce agents to invest effort strategically?
\newblock \emph{{ACM} Trans. Economics and Comput.}, 8\penalty0 (4):\penalty0
  19:1--19:23, 2020.
\newblock \doi{10.1145/3417742}.
\newblock URL \url{https://doi.org/10.1145/3417742}.

\bibitem[Kolotilin(2018)]{kolotilin2018optimal}
A.~Kolotilin.
\newblock Optimal information disclosure: A linear programming approach.
\newblock \emph{Theoretical Economics}, 13\penalty0 (2):\penalty0 607--635,
  2018.
\newblock \doi{https://doi.org/10.3982/TE1805}.
\newblock URL \url{https://onlinelibrary.wiley.com/doi/abs/10.3982/TE1805}.

\bibitem[Ku{\v{c}}ak et~al.(2018)Ku{\v{c}}ak, Juri{\v{c}}i{\'c}, and
  {\DJ}ambi{\'c}]{kuvcak2018machine}
D.~Ku{\v{c}}ak, V.~Juri{\v{c}}i{\'c}, and G.~{\DJ}ambi{\'c}.
\newblock Machine learning in education-a survey of current research trends.
\newblock \emph{Annals of DAAAM \& Proceedings}, 29, 2018.

\bibitem[Levanon and Rosenfeld(2021)]{levanon2021strategic}
S.~Levanon and N.~Rosenfeld.
\newblock Strategic classification made practical.
\newblock In M.~Meila and T.~Zhang, editors, \emph{Proceedings of the 38th
  International Conference on Machine Learning, {ICML} 2021, 18-24 July 2021,
  Virtual Event}, volume 139 of \emph{Proceedings of Machine Learning
  Research}, pages 6243--6253. {PMLR}, 2021.
\newblock URL \url{http://proceedings.mlr.press/v139/levanon21a.html}.

\bibitem[Mansour et~al.(2016)Mansour, Slivkins, Syrgkanis, and
  Wu]{MansourSSW16}
Y.~Mansour, A.~Slivkins, V.~Syrgkanis, and Z.~S. Wu.
\newblock Bayesian exploration: Incentivizing exploration in bayesian games.
\newblock In V.~Conitzer, D.~Bergemann, and Y.~Chen, editors, \emph{Proceedings
  of the 2016 {ACM} Conference on Economics and Computation, {EC} '16,
  Maastricht, The Netherlands, July 24-28, 2016}, page 661. {ACM}, 2016.
\newblock \doi{10.1145/2940716.2940755}.
\newblock URL \url{https://doi.org/10.1145/2940716.2940755}.

\bibitem[Mansour et~al.(2020)Mansour, Slivkins, and
  Syrgkanis]{mansour2015bayesian}
Y.~Mansour, A.~Slivkins, and V.~Syrgkanis.
\newblock Bayesian incentive-compatible bandit exploration.
\newblock \emph{Oper. Res.}, 68\penalty0 (4):\penalty0 1132--1161, 2020.
\newblock \doi{10.1287/opre.2019.1949}.
\newblock URL \url{https://doi.org/10.1287/opre.2019.1949}.

\bibitem[McDiarmid(1989)]{mcdiarmid1989method}
C.~McDiarmid.
\newblock \emph{On the method of bounded differences}, page 148–188.
\newblock London Mathematical Society Lecture Note Series. Cambridge University
  Press, 1989.
\newblock \doi{10.1017/CBO9781107359949.008}.

\bibitem[Raghavan et~al.(2020)Raghavan, Barocas, Kleinberg, and
  Levy]{raghavan2020mitigating}
M.~Raghavan, S.~Barocas, J.~M. Kleinberg, and K.~Levy.
\newblock Mitigating bias in algorithmic hiring: evaluating claims and
  practices.
\newblock In M.~Hildebrandt, C.~Castillo, L.~E. Celis, S.~Ruggieri, L.~Taylor,
  and G.~Zanfir{-}Fortuna, editors, \emph{FAT* '20: Conference on Fairness,
  Accountability, and Transparency, Barcelona, Spain, January 27-30, 2020},
  pages 469--481. {ACM}, 2020.
\newblock \doi{10.1145/3351095.3372828}.
\newblock URL \url{https://doi.org/10.1145/3351095.3372828}.

\bibitem[Rawal and Lakkaraju(2020)]{rawal2020beyond}
K.~Rawal and H.~Lakkaraju.
\newblock Beyond individualized recourse: Interpretable and interactive
  summaries of actionable recourses.
\newblock In H.~Larochelle, M.~Ranzato, R.~Hadsell, M.~Balcan, and H.~Lin,
  editors, \emph{Advances in Neural Information Processing Systems 33: Annual
  Conference on Neural Information Processing Systems 2020, NeurIPS 2020,
  December 6-12, 2020, virtual}, 2020.
\newblock URL
  \url{https://proceedings.neurips.cc/paper/2020/hash/8ee7730e97c67473a424ccfeff49ab20-Abstract.html}.

\bibitem[S{\'{a}}nchez{-}Monedero et~al.(2020)S{\'{a}}nchez{-}Monedero, Dencik,
  and Edwards]{sanchez2020does}
J.~S{\'{a}}nchez{-}Monedero, L.~Dencik, and L.~Edwards.
\newblock What does it mean to 'solve' the problem of discrimination in
  hiring?: social, technical and legal perspectives from the {UK} on automated
  hiring systems.
\newblock In M.~Hildebrandt, C.~Castillo, L.~E. Celis, S.~Ruggieri, L.~Taylor,
  and G.~Zanfir{-}Fortuna, editors, \emph{FAT* '20: Conference on Fairness,
  Accountability, and Transparency, Barcelona, Spain, January 27-30, 2020},
  pages 458--468. {ACM}, 2020.
\newblock \doi{10.1145/3351095.3372849}.
\newblock URL \url{https://doi.org/10.1145/3351095.3372849}.

\bibitem[Sellke and Slivkins(2021)]{sellke2021price}
M.~Sellke and A.~Slivkins.
\newblock The price of incentivizing exploration: {A} characterization via
  thompson sampling and sample complexity.
\newblock In P.~Bir{\'{o}}, S.~Chawla, and F.~Echenique, editors, \emph{{EC}
  '21: The 22nd {ACM} Conference on Economics and Computation, Budapest,
  Hungary, July 18-23, 2021}, pages 795--796. {ACM}, 2021.
\newblock \doi{10.1145/3465456.3467549}.
\newblock URL \url{https://doi.org/10.1145/3465456.3467549}.

\bibitem[Shavit et~al.(2020)Shavit, Edelman, and Axelrod]{shavit2020causal}
Y.~Shavit, B.~L. Edelman, and B.~Axelrod.
\newblock Causal strategic linear regression.
\newblock In \emph{Proceedings of the 37th International Conference on Machine
  Learning, {ICML} 2020, 13-18 July 2020, Virtual Event}, volume 119 of
  \emph{Proceedings of Machine Learning Research}, pages 8676--8686. {PMLR},
  2020.
\newblock URL \url{http://proceedings.mlr.press/v119/shavit20a.html}.

\bibitem[Tsirtsis and Rodriguez(2020)]{DBLP:conf/nips/TsirtsisR20}
S.~Tsirtsis and M.~G. Rodriguez.
\newblock Decisions, counterfactual explanations and strategic behavior.
\newblock In H.~Larochelle, M.~Ranzato, R.~Hadsell, M.~Balcan, and H.~Lin,
  editors, \emph{Advances in Neural Information Processing Systems 33: Annual
  Conference on Neural Information Processing Systems 2020, NeurIPS 2020,
  December 6-12, 2020, virtual}, 2020.
\newblock URL
  \url{https://proceedings.neurips.cc/paper/2020/hash/c2ba1bc54b239208cb37b901c0d3b363-Abstract.html}.

\bibitem[Ustun et~al.(2019)Ustun, Spangher, and Liu]{ustun2019actionable}
B.~Ustun, A.~Spangher, and Y.~Liu.
\newblock Actionable recourse in linear classification.
\newblock In danah boyd and J.~H. Morgenstern, editors, \emph{Proceedings of
  the Conference on Fairness, Accountability, and Transparency, FAT* 2019,
  Atlanta, GA, USA, January 29-31, 2019}, pages 10--19. {ACM}, 2019.
\newblock \doi{10.1145/3287560.3287566}.
\newblock URL \url{https://doi.org/10.1145/3287560.3287566}.

\bibitem[Wachter et~al.(2017)Wachter, Mittelstadt, and
  Russell]{wachter2017counterfactual}
S.~Wachter, B.~D. Mittelstadt, and C.~Russell.
\newblock Counterfactual explanations without opening the black box: Automated
  decisions and the {GDPR}.
\newblock \emph{CoRR}, abs/1711.00399, 2017.
\newblock URL \url{http://arxiv.org/abs/1711.00399}.

\end{thebibliography}
